\documentclass{amsart}[11pt]
\usepackage{amsmath, amsfonts, amsthm, amssymb,mathtools}
\usepackage{subfigure}
\usepackage{stmaryrd}
\usepackage{verbatim}
\usepackage{hyperref}
\usepackage{color}
\usepackage{ulem}
\usepackage[dvips]{graphicx}
\usepackage{enumerate}
\numberwithin{equation}{section}
\newtheorem{theorem}{Theorem}[section]
\newtheorem{corollary}[theorem]{Corollary}
\newtheorem{lemma}[theorem]{Lemma}

\theoremstyle{definition}

\newtheorem{remark}[theorem]{Remark}
\newtheorem{assumption}[theorem]{Assumption}

\newcommand{\ind}{1\hspace{-2.1mm}{1}} 
\newcommand{\I}{\mathtt{i}}
\newcommand{\EE}{\mathbb{E}}
\newcommand{\RR}{\mathbb{R}}
\newcommand{\QQ}{\mathbb{Q}}
\newcommand{\PP}{\mathbb{P}}
\newcommand{\D}{\mathrm{d}}

\newcommand{\E}{\mathrm{e}}

\newcommand{\sgn}{\mathrm{sgn}}
\newcommand{\Nn}{\mathcal{N}}
\newcommand{\Ss}{\mathcal{S}}

\newcommand{\Pp}{\mathcal{P}}
\newcommand{\Rr}{\mathcal{R}}
\newcommand{\Rf}{\mathfrak{R}}
\newcommand{\Ii}{\mathcal{I}}
\newcommand{\Hh}{\mathcal{H}}
\newcommand{\SVI}{\mathrm{SVI}}

\begin{document}

\title{Large-Maturity Regimes of the Heston Forward Smile}
\author{Antoine Jacquier and Patrick Roome}
\address{Department of Mathematics, Imperial College London}
\email{a.jacquier@imperial.ac.uk, p.roome11@imperial.ac.uk}
\date{\today}
\thanks{AJ acknowledges financial support from the EPSRC First Grant EP/M008436/1.
The authors would also like to thank an anonymous referee for useful comments.}
\maketitle

\begin{abstract}
We provide a full characterisation of the large-maturity forward implied volatility smile in the Heston model.
Although the leading decay is provided by a fairly classical large deviations behaviour, 
the algebraic expansion providing the higher-order terms highly depends on the parameters,
and different powers of the maturity come into play.
As a by-product of the analysis we provide new implied volatility asymptotics, both in the forward case 
and in the spot case, as well as extended SVI-type formulae.
The proofs are based on extensions and refinements of sharp large deviations theory,
in particular in cases where standard convexity arguments fail.
\end{abstract}


\section{Introduction}

Consider an asset price process $\left(\E^{X_t}\right)_{t\geq0}$ with $X_0=0$, paying no dividend, 
defined on a complete filtered probability space $(\Omega,\mathcal{F},(\mathcal{F}_t)_{t\geq0},\PP)$ 
with a given risk-neutral measure $\PP$, and assume that interest rates are zero. 
In the Black-Scholes-Merton (BSM) model, the dynamics of the logarithm of the asset price are given by
$\D X_t=-\frac{1}{2}\sigma^2  \D t +\sigma \D W_t$,
where $\sigma>0$ represents the instantaneous volatility and $W$ is a standard Brownian motion.
The no-arbitrage price of the call option at time zero is then given by the famous BSM formula~\cite{BS73,M73}:
$C_{\textrm{BS}}(\tau,k,\sigma) := \EE\left(\E^{X_\tau}-\E^k\right)_+
=\Nn\left(d_+\right)-\E^k\Nn\left(d_-\right)$,
with $d_\pm:=-\frac{k}{\sigma\sqrt{\tau}}\pm\frac{1}{2}\sigma\sqrt{\tau}$,
where $\Nn$ is the standard normal distribution function. 
For a given market price $C^{\textrm{obs}}(\tau,k)$ of the option at strike~$\E^k$ and maturity~$\tau$,
 the spot implied volatility $\sigma_{\tau}(k)$ is the unique solution to 
the equation $C^{\textrm{obs}}(\tau,k)=C_{\textrm{BS}}(\tau,k,\sigma_{\tau}(k))$.

For any~$t,\tau>0$ and $k\in\RR$, we define as in~\cite{B04,VL}
a  forward-start option with forward-start date~$t$, maturity~$\tau$ and strike~$\E^k$ 
as a European option with payoff 
$\left(\E^{X_{\tau}^{(t)}}-\E^k\right)^+$ where $\label{eq:XtTauDef}X_{\tau}^{(t)}:=X_{t+\tau}-X_t$ pathwise. 
By the stationary increment property, its value is simply~$C_{\textrm{BS}}(\tau,k,\sigma)$ in the BSM model. 
For a given market price~$C^{\textrm{obs}}(t,\tau,k)$ of the option at strike~$\E^k$, 
forward-start date~$t$ and maturity~$\tau$,
the forward implied volatility smile~$\sigma_{t,\tau}(k)$ is then defined (see also~\cite{B04})
 as the unique solution to 
$C^{\textrm{obs}}(t,\tau,k)=C_{\textrm{BS}}(\tau,k,\sigma_{t,\tau}(k))$. 
The forward smile is a generalisation of the spot implied volatility smile, and the two are equal when $t=0$. 

The literature on implied volatility asymptotics is extensive and has drawn upon a wide range of mathematical techniques.
Small-maturity asymptotics have received wide attention using heat kernel expansion results~\cite{Benarous}.
More recently, they have been studied using PDE methods~\cite{BBF,Hagan,Pascucci}, large deviations~\cite{DFJV,FJSmall}, saddlepoint methods~\cite{FJL11}, Malliavin calculus~\cite{Gobet,KT04} and differential geometry~\cite{GHLOW,PHL}.
Roger Lee~\cite{Lee} was the first to study extreme strike asymptotics, and further works on this have been carried out by Benaim and Friz~\cite{BF1, BF2} and in~\cite{Guli, GuliMass, GS, FGGS, DFJV, Marco}.
Large-maturity asymptotics have only been studied in~\cite{Tehr, FJ09, JMKR, JM12, FJM10} using large deviations and saddlepoint methods.
Fouque et al.~\cite{Fouque} have also successfully introduced perturbation techniques in order to study slow and fast mean-reverting stochastic volatility models.
Models with jumps (including L\'evy processes), studied in the above references for large maturities and extreme strikes, `explode' in small time, in a precise sense investigated in~\cite{Alos, AndLipton, Tank, Nutz, MijTankov, FL}.

On the other hand the literature on asymptotics of forward-start options and the forward smile is sparse.
Glasserman and Wu~\cite{GW10} use different notions of forward volatilities to assess their predictive values in determining future option prices and future implied volatility.
Keller-Ressel~\cite{KR} studies the forward smile asymptotic when the forward-start date $t$ becomes large ($\tau$ fixed).
Bompis~\cite{Bo13} produces an expansion for the forward smile  in local volatility models with bounded diffusion coefficient.
In~\cite{JR12} the authors compute small and large-maturity asymptotics for the forward smile in a general class of models (including stochastic volatility and time-changed exponential L\'evy models) where the forward characteristic function 
satisfies certain properties (in particular essential smoothness of the re-scaled limit).
In ~\cite{JR13} the authors prove that for fixed $t>0$ the Heston forward smile explodes as $\tau$ tends to zero.
Finally, empirical results on the forward smile have been carried out by practitioners in Balland~\cite{Ba}, Bergomi~\cite{B04}, B\"uhler~\cite{B02} and Gatheral ~\cite{G06}. 

Under some conditions on the parameters, it 
was shown in~\cite{JR12} that the smooth behaviour of the pointwise limit 
$\lim_{\tau\uparrow\infty}\tau^{-1}\log\EE(\E^{u X_{\tau}^{(t)}})$
yielded an asymptotic behaviour for the forward smile as 
$
\sigma_{t,\tau}^2(k\tau)=v_0^{\infty}(k)+v_1^{\infty}(k,t)\tau^{-1}+\mathcal{O}(\tau^{-2}),
$
where $v_0^{\infty}(\cdot)$ and $v_1^{\infty}(\cdot,t)$ are continuous functions on $\RR$.
In particular for $t=0$ (spot smiles), they recovered the result in~\cite{FJ09} (also under some restrictions on the parameters).
Interestingly, the limiting large-maturity forward smile $v_0^{\infty}$ does not depend on the forward-start date~$t$.
A number of practitioners  (see eg. Balland\cite{Ba}) have made the natural conjecture that the large-maturity forward smile should be the same as the large-maturity spot smile.
The result above rigorously shows us that this indeed holds if and only if the Heston correlation is close enough to zero.

It is natural to ask what happens when the parameter restrictions are violated.
We identify a number of regimes depending on the correlation and derive asymptotics in each regime.
The main results (Theorems~\ref{theorem:largematasympcalls} and~\ref{theorem:HestonLargeMatFwdSmileNonSteep}) state the following, as $\tau$ tends to infinity:
\begin{align*}
\EE\left(\E^{X^{(t)}_{\tau}}-\E^{k\tau}\right)^+
 & =
\Ii\left(k,\tau,V'(0),V'(1),\ind_{\{\kappa<\rho \xi\}}\right) + 
\frac{\phi(k,t)}{\tau^{\alpha}}\E^{-\tau\left(V ^*(k)-k \right)+\psi(k,t)\tau^\gamma}
\left(1+\mathcal{O}\left(\tau^{-\beta}\right)\right),\\
\sigma_{t,\tau}^2(k\tau)
 & =
v_0^{\infty}(k,t)+v_1^{\infty}(k,t)\tau^{-\lambda}+\Rr(\tau,\lambda),
\end{align*}
for any 
$k\in\RR$,
where $\Ii$ is some indicator function related to the intrinsic value of the option price, 
and~$\alpha$,~$\gamma$,~$\lambda$ strictly positive constants, depending on the level of the correlation.
The remainder~$\Rr$ decays to zero as $\tau$ tends to infinity.
If $t=0$ (spot smiles) we recover and extend the results in~\cite{FJ09}.

The paper is structured as follows. 
In Section~\ref{sec:largematreg} we introduce the different large-maturity regimes for the Heston model,
which will drive the asymptotic behaviour of forward option prices and forward implied volatilities.
In Section~\ref{sec:largematregfsaymp} we derive large-maturity forward-start option asymptotics in each regime and in Section~\ref{sec:largematregfsmaymp} we translate these results into forward smile asymptotics,
including extended SVI-type formulae (Section~\ref{sec:SVIExtended}).
Section~\ref{sec:nonsteepnumerics} provides numerics supporting the asymptotics developed in the paper and Section~\ref{sec:proofsnonsteep} gathers the proofs of the main results.

\textbf{Notations: } $\EE$ shall always denote expectation under a risk-neutral measure~$\PP$ given a priori.
We shall refer to the standard (as opposed to the forward) implied volatility as the spot smile and denote it $\sigma_\tau$.
The forward implied volatility will be denoted $\sigma_{t,\tau}$ as above
 and we let $\RR^*:=\RR\setminus\{0\}$ and $\RR^*_+:=(0,\infty)$.
For a sequence of sets $(\mathcal{D}_{ \varepsilon})_{\varepsilon>0}$ in $\RR$, 
we may, for convenience, use the notation $\lim_{\varepsilon \downarrow 0}\mathcal{D}_{ \varepsilon}$, 
by which we mean the following (whenever both sides are equal):
$
\liminf_{\varepsilon \downarrow 0}\mathcal{D}_{ \varepsilon}
:= \bigcup_{\varepsilon>0}\bigcap_{s\leq\varepsilon}\mathcal{D}_{s}
 = \bigcap_{\varepsilon>0}\bigcup_{s\leq\varepsilon}\mathcal{D}_{s}
=:\limsup_{\varepsilon\downarrow 0}\mathcal{D}_{\varepsilon}.
$
Finally, for a given set $A\subset\RR$, we let $A^o$ and $\overline{A}$ denote its interior and closure (in $\RR$),
 $\Re(z)$ and $\Im(z)$ the real and imaginary parts of a complex number~$z$,
and $\sgn(x) = 1$ if $x\geq 0$ and $-1$ otherwise.

\section{Large-maturity regimes}\label{sec:largematreg}

In this section we introduce the large-maturity regimes that will be used throughout the paper.
Each regime is determined by the Heston correlation and yields fundamentally different asymptotic behaviours for large-maturity forward-start options and the corresponding forward smile. 
This is due to the distinct behaviour of the moment explosions of the forward price process $(X_{\tau}^{(t)})_{\tau>0}$ in each regime . 
In the Heston model, the (log) stock price process is the unique strong solution to the following SDEs:
\begin{equation}\label{eq:Heston}
\begin{array}{rll}
\D X_t & = \displaystyle -\frac{1}{2}V_t\D t+ \sqrt{V_t}\D W_t, \quad & X_0=0,\\
\D V_t & = \kappa\left(\theta-V_t\right)\D t+\xi\sqrt{V_t}\D Z_t, \quad & V_0=v>0,\\
\D\left\langle W,Z\right\rangle_t & = \rho \D t,
\end{array}
\end{equation}
with $\kappa>0$, $\xi>0$, $\theta>0$ and $|\rho|<1$ and $(W_t)_{t\geq0}$ and $(Z_t)_{t\geq0}$ are two standard Brownian motions. 
We also introduce the notation $\mu:=2\kappa\theta/\xi^2$.
The Feller SDE for the variance process has a unique strong solution 
by the Yamada-Watanabe conditions~\cite[Proposition 2.13, page 291]{KS97}). 
The $X$ process is a stochastic integral of~$V$ and is therefore well-defined.  
The Feller condition, $2\kappa\theta\geq\xi^2$ (or $\mu\geq 1$), ensures that the origin is unattainable. 
Otherwise the origin is regular (hence attainable) and strongly reflecting
(see~\cite[Chapter 15]{KT81}). 
We however do not require the Feller condition in our analysis since we work with the forward moment generating function (mgf) of $X$.
Define the real numbers $\rho_-$ and $\rho_+$ by
\begin{equation}\label{eq:RhoPM}
\rho_{\pm} := \frac{\E^{-2 \kappa  t} \left(\xi(\E^{2\kappa t}-1) \pm(\E^{\kappa  t}+1) \sqrt{16 \kappa ^2 \E^{2 \kappa  t}+\xi ^2 (1-\E^{\kappa  t})^2}\right)}{8 \kappa },
\end{equation}
and note that $-1\leq \rho_-<0<\rho_+$ with $\rho_\pm=\pm 1$ if and only if $t=0$.
We now define the large-maturity regimes:
\begin{equation}\label{eq:Regimes}
\left.
\begin{array}{llr}
\Rf_1:  & \textit{Good correlation regime:} & \rho_{-}\leq\rho\leq\min(\rho_{+},\kappa/\xi); \\
\Rf_2:  & \textit{Asymmetric negative correlation regime:} & -1<\rho<\rho_{-}\text{ and }t>0;\\
\Rf_3:  & \textit{Asymmetric positive correlation regime:} & \rho_{+}<\rho<1\text{ and }t>0;\\
& \Rf_{3a}:  & \rho\leq \kappa/\xi;\\
& \Rf_{3b}:  & \rho>\kappa/\xi;\\
\Rf_4:  & \textit{Large correlation regime:} & 
\kappa/\xi<\rho\leq\min(\rho_{+},1).
\end{array}
\right.
\end{equation}
In the standard case $t=0$, $\Rf_1$ corresponds to $\kappa\geq \rho\xi$ and $\Rf_4$ is its complement.
We now define the following quantities:
\begin{equation}\label{eq:DefUpmU*pm}
u_{\pm}:=\frac{\xi-2\kappa\rho\pm\eta}{2\xi(1-\rho^2)}
\qquad \text{and}\qquad 
u_{\pm}^{*}:=\frac{\psi\pm\nu}{2 \xi (\E^{\kappa  t}-1)},
\end{equation}
with
\begin{equation}\label{eq:etanu}
\eta:=\sqrt{\xi^2(1-\rho^2)+(2\kappa-\rho\xi)^2},\qquad
\nu:=\sqrt{\psi^2-16\kappa^2\E^{\kappa t}}
\qquad \text{and}\qquad 
\psi:=\xi (\E^{\kappa  t}-1)-4\kappa\rho\E^{\kappa  t},
\end{equation}
as well as the interval $\mathcal{D}_\infty\subset\RR$ by
\begin{table}[h]\label{eq:DInfinityLargeMaturity}
  \begin{tabular}{| c || c | c | c | c | c | c |}
    \hline
&     $\Rf_1$ & $\Rf_2$ & $\Rf_{3a}$ & $\Rf_{3b}$ & $\Rf_4$\\ \hline
$\mathcal{D}_\infty$ & $[u_{-},u_{+}]$ & $[u_{-},u^*_{+})$ & 
$(u^*_{-},u_{+}]$ & $(u^*_{-},1]$ & $(u_{-},1]$
\\
\hline
  \end{tabular}
\end{table}

Note that for $t>0$,  $u_{+}>u_{+}^*>1$ if $\rho\leq\rho_{-}$
and $u_{-}<u_{-}^*<0$ if $\rho\geq\rho_{+}$
Also $\nu$ defined in~\eqref{eq:etanu} is a well-defined real number for all $\rho\in [-1,\rho_-]\cup[\rho_+,1]$.
Furthermore we always have $u_{-}<0$ and $u_{+}\geq 1$ with $u_{+}=1$ if and only if $\rho=\kappa/\xi$.
We define the real-valued functions $V$ and $H$ from $\mathcal{D}_{\infty}$ to $\RR$ by 
\begin{equation}\label{eq:VandH}
V(u)  := \frac{\mu}{2}\left(\kappa-\rho\xi u-d(u)\right)
\qquad\text{and}\qquad
H(u) :=
\frac{V(u)v \E^{-\kappa t}}{\kappa\theta -2 \beta _t V(u)}
-\mu\log\left(\frac{\kappa\theta-2 \beta _t V(u)}{\kappa\theta \left(1-\gamma \left(u\right)\right)}\right),
\end{equation}
with $d$, $\beta_t$ and $\gamma$ defined in~\eqref{eq:DGammaBeta}.
It is clear (see also~\cite{FJ09} and~\cite{JM12}) that the function~$V$
is infinitely differentiable, strictly convex and essentially smooth on the open interval $\left(u_{-},u_{+}\right)$ and that $V(0)=0$.
Furthermore $V(1)=0$ if and only if $\rho\leq\kappa/\xi$.
For any $k\in\RR$ the (saddlepoint) equation $V'(u^*(k))=k$ has a unique solution~$u^*(k)\in(u_-,u_+)$:
\begin{align}\label{eq:u*}
u^*(k):=\frac{\xi-2\kappa\rho+\left(\kappa\theta\rho+k\xi\right)\eta\left(k^2\xi^2+2k\kappa\theta\rho\xi+\kappa^2\theta^2\right)^{-1/2}}{2\xi \left(1-\rho^2\right)}.
\end{align}
Further let $V^*:\RR\to\RR_+$ denote the Fenchel-Legendre transform of $V$:
\begin{equation}\label{eq:VStarDefinition}
V^*(k):=\sup_{u\in\mathcal{D}_{\infty}}\left\{uk-V(u)\right\},
\qquad\text{for all }k\in\RR.
\end{equation}

The following lemma characterises $V^*$ and can be proved using straightforward calculus.
As we will see in Section~\ref{sec:LDP}, the function $V^*$ can be interpreted as a large deviations rate function for our problem.

\begin{lemma}\label{lemma:V*Characterisation}
Define the function $W(k,u)\equiv u k-V(u)$ for any $(k,u)\in\RR\times [u_-,u_+]$.
Then
\begin{itemize}
\item $\Rf_1$:
$V^*(k) \equiv W(k,u^*(k))$ on $\RR$;
\item $\Rf_2$:
$V^*(k) \equiv W(k,u^*(k))$ on $(-\infty, V'(u_+^*)]$ and 
$V^*(k) \equiv W(k,u^*_+)$ on $(V'(u^*_{+}), +\infty)$;
\item $\Rf_{3a}$:
$V^*(k) \equiv W(k,u^*_-)$ on $(-\infty, V'(u^*_{-}))$ and 
$V^*(k) \equiv W(k,u^*(k))$ on $[ V'(u_-^*), +\infty)$;
\item $\Rf_{3b}$:
$$
V^*(k) \equiv
\left\{
\begin{array}{ll}
W(k,u^*_-), \quad & \text{on } (-\infty, V'(u^*_-)),\\
W(k,u^*(k)), \quad & \text{on } [V'(u^*_-), V'(1)],\\
W(k,1), \quad & \text{on } (V'(1), +\infty);
\end{array}
\right.
$$
\item $\Rf_4$:
$V^*(k) \equiv W(k,u^*(k))$ on $(-\infty, V'(1)]$ and 
$V^*(k) \equiv W(k,1)$ on $(V'(1),+\infty)$.
\end{itemize}
\end{lemma}

\section{Forward-start option asymptotics}\label{sec:largematregfsaymp}

In order to specify the forward-start option asymptotics we need to introduce some functions and constants.
As outlined in Theorem~\ref{theorem:largematasympcalls}, 
each of them is defined in a specific regime and strike region where it is well defined and real valued. 
In the formulae below, $\gamma$, $\beta_t$ are defined in~\eqref{eq:DGammaBeta}, $u^*_{\pm}$ in~\eqref{eq:DefUpmU*pm} and $V$ in~\eqref{eq:VandH}. 
\begin{equation}\label{eq:a}
\left\{
\begin{array}{rlrl}
a_1^{\pm}(k) & := \displaystyle \mp \frac{2|k-V'(u_{\pm}^*)|}{\zeta^2_{\pm}(k)},
& a_2^{\pm}(k) & := \displaystyle \frac{\mu\E^{-\kappa t}}{16\beta _t^2}
\frac{\xi ^2 v V''(u_{\pm}^*)-8 \beta _t^2 \E^{\kappa  t} V'(u_{\pm}^*) \left(k-V'(u_{\pm}^*)\right)}{V'(u_{\pm}^*)\left(k-V'(u_{\pm}^*)\right)^2}, \\
\widetilde{a}_1^{\pm} & := \displaystyle \mp\left|\frac{\E^{-\kappa t}\kappa\theta v}{4 V'(u_{\pm}^*)V''(u_{\pm}^*)\beta_t^2}\right|^{1/3},
& \widetilde{a}_2^{\pm}  & :=\displaystyle -\frac{(\kappa\theta\E^{-\kappa t})^{2/3}}{12 \xi ^2 v^{1/3}\beta _t^{4/3} } \frac{16 V'(u_{\pm}^*) V''(u_{\pm}^*) \beta _t^2 \E^{\kappa  t}+\xi ^2 v V'''(u_{\pm}^*)}{2^{1/3} |V'(u_{\pm}^*)|^{2/3} V''(u_{\pm}^*)^{5/3}},
\end{array}
\right.
\end{equation}
where
\begin{equation}\label{eq:variance}
\zeta^2_{\pm}(k)
 := \displaystyle 4 \beta _t \left(\frac{V'(u_{\pm}^*)(k-V'(u_{\pm}^*))^3}{\kappa\theta v \E^{-\kappa t}}\right)^{1/2}; 
\end{equation}
\begin{equation}\label{eq:e}
\left\{
\begin{array}{rlrl}
e_0^{\pm}(k) & := \displaystyle -2\beta_t a_1^{\pm}(k)V'(u_{\pm}^*),
&  e_1^{\pm}(k) & :=-\beta_t \left[V''(u_{\pm}^*)a_1^{\pm}(k)^2+2V'(u_{\pm}^*) a_2^{\pm}(k)\right], \\
\widetilde{e}_0^{\pm} & := \displaystyle  -2\beta_t \widetilde{a}_1^{\pm}V'(u_{\pm}^*),
& \widetilde{e}_1^{\pm} & := -\beta_t \left[V''(u_{\pm}^*)(\widetilde{a}_1^{\pm})^2+2V'(u_{\pm}^*) \widetilde{a}_2^{\pm}\right], 
\end{array}
\right.
\end{equation}
\begin{equation}\label{eq:c0}
\left\{
\begin{array}{ll}
c_0^{\pm}(k) := \displaystyle -2a_1^{\pm}(k)\left(k-V'(u_{\pm}^*)\right),
\quad 
 c_2^{\pm}(k) :=\displaystyle \left(\frac{\kappa\theta\left(1-\gamma(u_{\pm}^*) \right) }{e_0^{\pm}(k)}\right)^{\mu}, & \\
c_1^{\pm}(k) := \displaystyle   v \E^{-\kappa  t} \left(\frac{a_1^{\pm}(k) V'(u_{\pm}^*)}{e_0^{\pm}(k)}-\frac{\kappa\theta e_1^{\pm}(k)}{2 e_0^{\pm}(k)^2 \beta _t}\right)-a_2^{\pm}(k) \left(k-V'(u_{\pm}^*)\right)+\frac{1}{2}
   a_1^{\pm}(k)^2 V''(u_{\pm}^*),& 
\end{array}
\right.
\end{equation}

\begin{equation}\label{eq:c0bound}
\left\{
\begin{array}{ll}
\widetilde{c}_0^{\pm} := \displaystyle \frac{3}{2} (\widetilde{a}_1^{\pm})^2 V''(u_{\pm}^*), 
\quad 
 \widetilde{c}_2^{\pm} := \left(\frac{\kappa\theta\left(1-\gamma(u_{\pm}^*) \right)}{\widetilde{e}_0^{\pm}}\right)^{\mu},
\quad
 g_0 := \frac{v \E^{-\kappa t} V(1)}{\kappa\theta -2 \beta _t V(1)}, & \\
\widetilde{c}_1^{\pm} := \displaystyle   v \E^{-\kappa  t} \left(\frac{\widetilde{a}_1^{\pm} V'(u_{\pm}^*)}{\widetilde{e}_0^{\pm}}-\frac{\kappa\theta\widetilde{e}_1^{\pm}}{2 (\widetilde{e}_0^{\pm})^2 \beta _t}\right)+\widetilde{a}_1^{\pm}\widetilde{a}_2^{\pm} V''(u_{\pm}^*) +\frac{(\widetilde{a}_1^{\pm})^3 V'''(u_{\pm}^*)}{6},& 
\end{array}
\right.
\end{equation}
\begin{equation}\label{eq:phi0largetime}
\phi_0(k)
 := \displaystyle \frac{1}{\sqrt{2\pi V''(u^*(k))}} 
\left\{
\begin{array}{ll}
\displaystyle  \frac{\exp\left({H(u^*(k))}\right)}{u^*(k) (u^*(k)-1)},
& \text{if }k\in\RR\setminus\{V'(0),V'(1)\},\\
\displaystyle  \left(-1-\sgn(k)\left(\frac{V'''(u^*(k))}{6 V''(u^*(k))}-H'(u^*(k))\right)\right),
& \text{if }k\in\{V'(0),V'(1)\}.
\end{array}
\right.
\end{equation}
\begin{equation}
\left\{
\begin{array}{llrl}
\phi_\pm (k) & := \displaystyle  \frac{c_2^\pm(k)\E^{c_1^\pm(k)}}{\zeta_{\pm}(k)u^*_\pm(u^*_\pm-1)\sqrt{2\pi}},
& \widetilde{\phi}_\pm & := \displaystyle \frac{\tilde{c}_2^\pm\E^{\tilde{c}_1^\pm}}{u^*_\pm(u^*_\pm-1)\sqrt{6\pi V''(u^*_{\pm})}},\\
\phi_2 (k) & := \displaystyle \frac{-\E^{g_0}}{\Gamma(1+\mu)}
\left(\frac{2\mu (\kappa-\rho\xi)^2 (k-V'(1))}{\kappa\theta-2\beta_t V(1)}\right)^{\mu},
& \phi_1 & := \displaystyle \frac{-\E^{g_0}}{2\Gamma (1+\mu/2)}
\left(\frac{\mu(\kappa-\rho\xi)^2 \sqrt{2 V''(1) }}{\kappa\theta-2\beta_t V(1)}\right)^{\mu},
\end{array}
\right.
\end{equation}

Since $u^*_-<0$ and $u^*_+>1$, we always have $V'(u^*_+)>0$ and $V'(u^*_-)<0$.
Furthermore, $V''(u^*_{\pm})>0$ and one can show that $\gamma(u^*_{\pm})\neq1$;
therefore all the functions and constants in ~\eqref{eq:a},~\eqref{eq:variance},~\eqref{eq:e},~\eqref{eq:c0} and ~\eqref{eq:c0bound} are well-defined and real-valued.
$\phi_0$ is well-defined since $V''(u^*(k))>0$ and $\phi_2$ and the constant $\phi_1$ are well-defined since $\kappa\theta-2\beta_t V(1)>0$.
Finally define the following combinations and the function $\Ii:\RR\times\RR_{+}^*\times\RR^3\to\RR$ :
\begin{equation}
\begin{array}{llllll}
\Hh_0: 
 & \alpha = \frac{1}{2}, & \beta = 1, & \gamma = 0, & \phi \equiv \phi_0, & \psi \equiv 0,
 \\
\widetilde{\Hh}_\pm: 
 & \alpha = \frac{\mu}{3}-\frac{1}{2}, & \beta = \frac{1}{3}, & \gamma = \frac{1}{3}, 
 & \phi \equiv \widetilde{\phi}_{\pm}, & \psi \equiv \tilde{c}_0^{\pm},
 \\
\Hh_\pm: 
 & \alpha = \frac{\mu}{2}-\frac{3}{4}, & \beta = \frac{1}{2}, & \gamma = \frac{1}{2}, & \phi \equiv \phi_{\pm}, 
 & \psi \equiv c_0^{\pm}, \\
\Hh_1: 
 & \alpha = -\frac{\mu}{2}, & \beta = \frac{1}{2}, & \gamma = 0, & \phi \equiv \phi_{1}, & \psi \equiv 0, 
 \\
\Hh_2: 
 & \alpha = -\mu, & \beta = 1, & \gamma = 0, & \phi \equiv \phi_{2}, & \psi \equiv 0, 
\end{array}
\end{equation}
\begin{equation}\label{eq:intrinsiclargetimenonsteep}
\Ii(k,\tau,a,b,c)
:= \left(1-\E^{k \tau}\right)\ind_{\{k<a\}}+\ind_{\{a<k<b\}}
+c\ind_{\{b\leq k\}}
+\frac{1-c}{2}\ind_{\{k=b\}}
+\left(1-\frac{1}{2}\E^{k \tau}\right)\ind_{\{k=a\}}.
\end{equation}

We are now in a position to state the main result of the paper, 
namely an asymptotic expansion for forward-start option prices in all regimes for all (log) strikes on the real line.
The proof is obtained using Lemma~\ref{lem:nonsteeplem} in conjunction with the asymptotics in Lemmas~\ref{lemma:V*Asymptotics},~\ref{lemma:FourierAsymptotics},~\ref{lem:v0v1F} and~\ref{lem:v0v1}.

\begin{theorem}\label{theorem:largematasympcalls}
The following expansion holds for forward-start call options for all $k\in\RR$ as $\tau$ tends to infinity:
$$
\EE\left(\E^{X^{(t)}_{\tau}}-\E^{k\tau}\right)^+
=
\Ii\left(k,\tau,V'(0),V'(1),\ind_{\{\kappa<\rho \xi\}}\right) + 
\frac{\phi(k,t)}{\tau^{\alpha}}\E^{-\tau\left(V ^*(k)-k \right)+\psi(k,t)\tau^\gamma}
\left(1+\mathcal{O}\left(\tau^{-\beta}\right)\right),
$$
where the functions $\phi$, $\psi$ and the constants $\alpha$, $\beta$ and $\gamma$ are given by the following combinations\footnote{whenever $\Hh_0$ is in force, the case $k=V'(a)$ 
is excluded if $v = \theta \Upsilon(a)$, with $\Upsilon$ defined in~\eqref{eq:Upsilon}, for $a\in\{0,1\}$.}:
\begin{itemize}
\item $\Rf_1$: $\Hh_0$ for $k\in\RR$;

\item $\Rf_2$:
$\Hh_0$ for $k\in(-\infty, V'(u_+^*))$;
$\widetilde{\Hh}_+$ for $k=V'(u_+^*)$;
$\Hh_+$ for $k\in(V'(u_+^*),+\infty)$;

\item $\Rf_{3a}$:
$\Hh_-$ for $k\in(-\infty, V'(u_-^*))$;
$\widetilde{\Hh}_-$ for $k=V'(u_-^*)$;
$\Hh_0$ for $k\in(V'(u_-^*),+\infty)$;

\item $\Rf_{3b}$:
$\Hh_-$ for $k\in(-\infty, V'(u_-^*))$;
$\widetilde{\Hh}_-$ for $k=V'(u_-^*)$;
$\Hh_0$ for $k\in(V'(u_-^*),V'(1))$;
$\Hh_1$ at $k=V'(1)$;
$\Hh_2$ for $k\in(V'(1),+\infty)$;

\item $\Rf_4$:
$\Hh_0$ for $k\in(-\infty, V'(1))$;
$\Hh_1$ for $k=V'(1)$;
$\Hh_2$ for $k\in(V'(1),+\infty)$;
\end{itemize}
\end{theorem}

In order to highlight the symmetries appearing in the asymptotics, 
we shall at times identify an interval with the corresponding regime and combination in force.
This slight abuse of notations should not however be harmful to the comprehension.

\begin{remark}\label{rem:largematfwdstartoption}\
\begin{enumerate}[(i)]
\item Under~$\Rf_1$, asymptotics for the large-maturity forward smile 
(for $k\in\mathbb{R}\setminus\{V'(0),V'(1)\}$)
have been derived in~\cite[Proposition 3.8]{JR12}.
\item For $t=0$, large-maturity asymptotics have been derived in~\cite{FJL11} under~$\Rf_1$
and partially in~\cite{JM12} under~$\Rf_4$.
\item 
All asymptotic expansions are given in closed form and can in principle be extended to arbitrary order.
\item When $\Hh_{\pm}$ and $\Hh_2$ are in force then
$V^*(k)-k$ is linear in $k$ as opposed to being strictly convex as in $\Hh_0$.
\item 
If $\rho\leq\kappa/\xi$ then $V^*(k)-k\geq0$ with equality if and only if $k=V'(1)$.
If $\rho>\kappa/\xi$ then  $V^*(k)-k\geq -V(1)>0$.
Since $\gamma\in[0,1)$, the leading order decay term is given by $\E^{-\tau\left(V ^*(k)-k\right)}$.
\item
Under $\Hh_2$ (which only occur when $\rho>\kappa/\xi$ for log-strikes strictly greater than $V'(1)$), 
forward-start call option prices decay to one as $\tau$ tends to infinity.
This is fundamentally different than the large-strike behaviour in other regimes and in the BSM model, 
where call option prices decay to zero.
This seemingly contradictory behaviour is explained as follows:
as the maturity increases there is a positive effect on the price by an increase in the time value of the option and a negative effect on the price by increasing the strike of the forward-start call option.
In standard regimes and for sufficiently large strikes the strike effect is more prominent than the time value effect in the large-maturity limit.
Here, because of the large correlation, this effect is opposite: 
as the asset price increases, the volatility tends to increase driving the asset price to potentially higher levels. 
This gamma or time value effect outweighs the increase in the strike of the option. 
\item
In~$\Rf_4$, the decay rate $V^*(k)-k$ has a very different behaviour: 
the minimum achieved at $V'(1)$ is not zero and  $V^*(k)-k$ is constant for $k\geq V'(1)$.
There is limited information in the leading-order behaviour 
and important distinctions must therefore occur in higher-order terms.
This is illustrated in Figures~\ref{fig:LargeCorrelRegime10years} and~\ref{fig:LargeCorrelRegime20years} where the first-order asymptotic is vastly superior to the leading order.
\item
It is important to note that  $u_{\pm}^*$ and $V^*$ depend on the forward-start date~$t$ through~\eqref{eq:DefUpmU*pm} 
and the regime choice. 
However, in the uncorrelated case $\rho=0$, $\Rf_1$ always applies and $V^*$ does not depend on~$t$.
The non-stationarity of the forward smile over the spot smile (at leading order) depends critically on how far the correlation is away from zero.
\end{enumerate}
\end{remark}

In order to translate these results into forward smile asymptotics (in the next section), 
we require a similar expansion for the Black-Scholes model, 
where the log stock price process satisfies
$\D X_t = -\frac{1}{2}\Sigma^2 \D t + \Sigma\D W_t$, with $\Sigma>0$. 
Define the functions $V^*_{\textrm{BS}}:\RR\times\RR^*_{+}\to\RR$ and $\phi_{\textrm{BS}}:\RR\times\RR^*_{+}\times\RR\to\RR$ by $V^*_{\textrm{BS}}(k,a):=
\left(k+a/2\right)^2/(2 a)$ and
$$
\phi_{\textrm{BS}}(k,a,b) \equiv
\frac{4a^{3/2}}{(4k^2-a^2)\sqrt{2\pi}}\exp\left(b\left(\frac{k^2}{2a^2}-\frac{1}{8}\right)\right)
\ind_{\{k\neq\pm a/2\}}
+\frac{b-2}{2\sqrt{2 a \pi}}
\ind_{\{k=\pm a/2\}},
$$
so that the following holds (see~\cite[Corollary 2.11]{JR12}):
\begin{lemma}\label{Cor:BSOptionLargeTime}
Let $a>0$, $b\in\RR$ and set $\Sigma^2:=a+b/\tau$ for $\tau$ large enough so that $a+b/\tau>0$.
In the BSM model the following expansion then holds for any $k\in\RR$ as $\tau$ tends to infinity
(the function~$\Ii$ is defined in\eqref{eq:intrinsiclargetimenonsteep}):
$$
\EE\left(\E^{X^{(t)}_{\tau}}-\E^{k\tau}\right)^+
=
\Ii\left(k,\tau,-\frac{a}{2},\frac{a}{2},0\right)
+
\frac{\phi_{\textrm{BS}}(k,a,b) }{\tau^{1/2}}\E^{-\tau\left(V^*_{\textrm{BS}}(k,a)-k\right)}\left(1+\mathcal{O}(\tau^{-1})\right).
$$
\end{lemma}

\subsection{Connection with large deviations}\label{sec:LDP}
Although clear from Theorem~\ref{theorem:largematasympcalls}, we have so far not mentioned
the notion of large deviations at all.
The leading-order decay of the option price as the maturity tends to infinity
gives rise to estimates for large-time probabilities;
more precisely, by formally differentiating both sides with respect to the log-strike, 
one can prove, following a completely analogous proof to~\cite[Corollary 3.3]{JR13}, that
$$
-\lim_{\tau\uparrow\infty}\tau^{-1}\log\PP\left(X^{(t)}_{\tau} \in B\right) = \inf_{z\in B} V^*(z),
$$
for any Borel subset~$B$ of the real line, namely that
$(X^{(t)}_{\tau}/\tau)_{\tau>0}$ satisfies a large deviations principle under~$\PP$ 
with speed~$\tau$ and good rate function $V^*$ as~$\tau$ tends to infinity.
We refer the reader to the excellent monograph~\cite{DZ93} for more details on large deviations.
The theorem actually states a much stronger result here since it provides higher-order estimates,
coined `sharp large deviations' in~\cite{BR02} (see also~\cite[Definition 1.1]{Bercu}).
Now, classical methods to prove large deviations, when the the moment generating function is known
rely on the G\"artner-Ellis theorem.
In mathematical finance, one can consult for instance~\cite{FJSmall},~\cite{FJ09} or~\cite{JMKR}
for the small-and large-time behaviour of stochastic volatility models,
and~\cite{Pham} for an insightful overview.
The G\"artner-Ellis theorem requires, in particular, the limiting logarithmic moment generating function~$V$
to be steep at the boundaries of its effective domain.
This is indeed the case in Regime $\Rr_1$, but fails to hold in other regimes.
The standard proof of this theorem (as detailed in~\cite[Chapter 2, Theorem 2.3.6]{DZ93})
clearly holds in the open intervals of the real line where the function $V$ is strictly convex,
encompassing basically all occurrences of $\Hh_0$.
The other cases, when~$V^*$ becomes linear, and the turning points $V'(0)$ and $V'(1)$,
however have to be handled with care and solved case by case.
Proving sharp large deviations essentially relies on finding a new probability measure under which
a rescaled version of the original process converges weakly to some random variable 
(often Gaussian, but not always);
in layman terms, under this new probability measure, the rare events / large deviations 
of the rescaled variable are not rare any longer.
More precisely, fix some log-moneyness~$k\in\RR$;
we determine a process $(Z_{\tau,k})_{\tau>0} := ((X^{(t)}_{\tau,k}-k\tau)/\tau^{\alpha})_{\tau>0}$,
and a probability measure $\QQ_{k,\tau,\alpha}$ via
$$
\frac{\D\QQ_{k,\tau,\alpha}}{\D\PP}:=\exp\left( u^*_{\tau}(k)X_{\tau}^{(t)}-\tau\Lambda^{(t)}_{\tau}(u^*_{\tau}(k))\right),
$$
where $u^*_{\tau}(k)$ is the unique solution to the equation 
$\partial_{u}\Lambda^{(t)}_{\tau}(u_{\tau}^*(k))=k$,
with~$\Lambda^{(t)}_{\tau}$ denoting the (rescaled) logarithmic moment generating function of~$X_\tau^{(t)}$
(See Section~\ref{sec:FwdLimits}).
The characteristic function $\Phi_{\tau,k,\alpha}(u):=\EE^{\QQ_{k,\tau,\alpha}}(\E^{\I u Z_{\tau,k}})$ has some expansion
as~$\tau$ tends to infinity.
Once this pair has been found, the final part of the proof is to express call prices (or probabilities)
as inverse Fourier transforms of the characteristic function multiplied by some kernel,
and to use the expansion of~$\Phi_{\tau,k,\alpha}(u)$ to determine the desired asymptotics.
The main technical issues, and where the different regimes come into play, 
arise in the properties of the asymptotic behaviour of~$\Phi_{\tau,k,\alpha}(u)$ 
and~$u^*_\tau(k)$ as~$\tau$ tends to infinity
(and on the value~$\alpha$ one has to choose).
More precise details about the main steps of the proofs 
are provided at the beginning of Section~\ref{sec:proofsnonsteep} 
and in Section~\ref{sec:Genmethlargetime}.

Sharp large deviations, or more generally speaking, probabilistic asymptotic expansions,
\`a la Bahadur-Rao~\cite{Bahadur}, can also be proved via other routes.
In particular, the framework developed by Benarous~\cite{Benarous} 
(and applied to the financial context in~\cite{DFJV1, DFJV}) is an extremely powerful
tool to handle Laplace methods on Wiener space and heat kernel expansions.
However, the singularity of the square root diffusion (in the SDE~\eqref{eq:Heston} for the variance) 
at the origin falls outside the scope of this theory.
Incidentally, Conforti, Deuschel and De Marco~\cite{CDDM} recently
proved a (sample path) large deviations principle for the square root diffusion, 
giving hope for an alternative proof to ours. 
As explained in~\cite{JR13, JR15}, the forward-start framework on the couple~$(X_\tau,V_\tau)_{\tau\geq 0}$,
solution to~\eqref{eq:Heston}, starting at $(0,v)$,
can be seen as a standard option pricing problem on the forward couple
$(X_{\tau}^{(t)}, V_{\tau}^{(t)})_{\tau\geq 0}$, solution to the same stochastic differential equation~\eqref{eq:Heston},
albeit starting at the point $(0, V_t)$, namely with random initial variance.
This additional layer of complexity arising from starting the SDE 
at a random starting point makes the application of the Benarous framework 
as well as the Conforti-Deuschel-De Marco result,  a fascinating, yet challenging, exercise to consider.

\section{Forward smile asymptotics}~\label{sec:largematregfsmaymp}

We now translate the forward-start option asymptotics obtained above into asymptotics of the forward implied volatility smile. 
Let us first define the function
$v_0^{\infty}:\RR\times\RR_{+}\to\RR$ by
\begin{equation}\label{eq:v0nonsteep}
v_0^{\infty}(k,t) := 2 \left(2 V^*(k) -k+2 \mathcal{Z}(k) \sqrt{V^*(k)(V^*(k)-k)}\right),\quad\text{for all }k\in\RR, t\in\RR_+
\end{equation}
with $\mathcal{Z}:\RR\to\{-1,1\}$ defined by
$
\mathcal{Z}(k) \equiv 
\ind_{\{k\in[V'(0),V'(1)]\}} +\sgn(\rho\xi-\kappa)\ind_{\{k>V'(1)\}}-\ind_{\{k<V'(0)\}}
$
and $V^*$ given in Lemma~\ref{lemma:V*Characterisation}.
Define the following combinations:
$$
\begin{array}{llllll}
\Pp_0: 
 & \chi \equiv \chi_0, \quad & \eta \equiv 1, \quad & \lambda =1, \quad & \Rr(\tau,\lambda) = \mathcal{O}(\tau^{-2\lambda}),\\
\widetilde{\Pp}_\pm: 
 & \chi \equiv \widetilde{c}_0^{\pm}, \quad & \eta \equiv 1, \quad & \lambda = \frac{2}{3}, 
 \quad & \Rr(\tau,\lambda) = o(\tau^{-\lambda}),\\
\Pp_\pm: 
 & \chi \equiv c_0^{\pm}, \quad & \eta \equiv 1, \quad & \lambda = \frac{1}{2}, 
 \quad & \Rr(\tau,\lambda)=
\left\{
\begin{array}{ll}
\displaystyle o\left(\tau^{-\lambda}\right), 
& \text{if } \mu\ne 1/2,\\
\displaystyle \mathcal{O}\left(\tau^{-2\lambda}\right),
& \text{if } \mu = 1/2,
\end{array}
\right.\\
\Pp_1: 
 & \chi \equiv0, \quad & \eta \equiv 0, \quad & \lambda = 0, 
 \quad & \Rr(\tau,\lambda)= o(1).
\end{array}
$$
Here  $c_0^{\pm}$ and $\widetilde{c}_0^{\pm}$ are given in~\eqref{eq:c0} and ~\eqref{eq:c0bound} and
$\chi_0:\RR\setminus\{V'(0),V'(1)\}\to\RR$ is defined by
\begin{equation}\label{eq:chi0largetime}
\chi_0(k,t) \equiv 
H(u^*(k))+
\log \left(\frac{4 k^2-v_0^{\infty}(k,t)^2}{4 (u^*(k)-1) u^*(k) v_0^{\infty}(k,t)^{3/2} \sqrt{V''(u^*(k))}}\right),
\end{equation}
with $V$ and $H$ given in~\eqref{eq:VandH} and $u^*$ in~\eqref{eq:u*}. 
We now state the main result of the section, namely an expansion for the forward smile in all regimes and (log) strikes on the real line.
The proof is given in Section~\ref{sec:fwdsmileproofnonsteeplargemat}.
\begin{theorem}\label{theorem:HestonLargeMatFwdSmileNonSteep}
The following expansion holds for the forward smile as $\tau$ tends to infinity:
$$
\sigma_{t,\tau}^2(k\tau)
=
v_0^{\infty}(k,t)+v_1^{\infty}(k,t)\tau^{-\lambda}+\Rr(\tau,\lambda),
\qquad\text{for any }k\in\RR,
$$
where $v_1^{\infty}:\RR\times\RR_{+}\to\RR$  is defined by
\begin{equation*}
v_1^{\infty}(k,t)
 := 
\left\{
\begin{array}{ll}
\displaystyle\frac{8 v_0^{\infty}(k,t)^2}{4 k^2-v_0^{\infty}(k,t)^2}\chi(k,t),
\quad & \text{if }k\in\RR\setminus\{V'(0),V'(1)\},\\
\displaystyle 2\eta(k)\left[1-\sqrt{\frac{v_0^{\infty}(k,t)}{V''(u^*(k))}}\left(1+\sgn(k)\left(\frac{V'''(u^*(k))}{6 V''(u^*(k))}-H'(u^*(k))\right)\right)\right],
\quad & \text{if }k\in\{V'(0),V'(1)\},
\end{array}
\right.
\end{equation*}
with the functions $\chi,\eta$, the remainder $\Rr$ and the constant $\lambda$ given by the following 
combinations\footnote{whenever $\Pp_0$ is in force, the case $k=V'(a)$ 
is excluded if $v = \theta \Upsilon(a)$, with $\Upsilon$ defined in~\eqref{eq:Upsilon}, for $a\in\{0,1\}$.}:
\begin{itemize}
\item $\Rf_1$: $\Pp_0$ for $k\in\RR$;
\item $\Rf_2$:
$\Pp_0$ for $k\in(-\infty, V'(u_+^*))$;
$\widetilde{\Pp}_+$ for $k=V'(u_+^*)$;
$\Pp_+$ for $k\in(V'(u_+^*),+\infty)$;

\item $\Rf_{3a}$:
$\Pp_-$ for $k\in(-\infty, V'(u_-^*))$;
$\widetilde{\Pp}_-$ for $k=V'(u_-^*)$;
$\Pp_0$ for $k\in(V'(u_-^*),+\infty)$;

\item $\Rf_{3b}$:
$\Pp_-$ for $k\in(-\infty, V'(u_-^*))$;
$\widetilde{\Pp}_-$ for $k=V'(u_-^*)$;
$\Pp_0$ for $k\in(V'(u_-^*),V'(1))$;
$\Pp_1$ for $k\in[V'(1),+\infty)$;

\item $\Rf_4$:
$\Pp_0$ for $k\in(-\infty, V'(1))$;
$\Pp_1$ for $k\in[V'(1),+\infty)$.
\end{itemize}
\end{theorem}

\begin{remark} \ \label{rem:contlargemat}
\begin{enumerate}[(i)] 
\item 
In the standard spot case $t=0$, the large-maturity asymptotics of the implied volatility smile was derived 
in~\cite{{FJM10}} for $\Rf_1$ only (i.e. assuming $\kappa> \rho\xi$).
In the complementary case, $\Rf_4$, the behaviour of the smile for large strikes become more degenerate, and one cannot specify higher-order asymptotics for $k \geq V'(1)$.
\item
The zeroth-order term $v_0^{\infty}$ is continuous on $\RR$ (see also section~\ref{sec:SVIExtended}), 
which is not necessarily true for higher-order terms.
In $\Rf_2$, $\Rf_{3a}$ and $\Rf_{3b}$, 
$v_1^{\infty}$ tends to either infinity or zero at the critical strikes $V'(u^*_{+})$ and $V'(u^*_{-})$
(this is discussed further in Section~\ref{sec:nonsteepnumerics}).
In $\Rf_1$, $v_1^{\infty}$ is continuous on the whole real line.
\item 
Straightforward computations show that $0<v_0^{\infty}(k)<2|k|$ for $k\in\RR\setminus\left[V'(0), V'(1)\right]$, 
and $v_0^{\infty}(k)>2|k|$ for $k\in\left(V'(0), V'(1)\right)$, 
so that $v_1^\infty$ is well defined on $\RR\setminus\{V'(0), V'(1)\}$.
On $(-\infty, V'(u_-^*))\cup(V'(u_+^*),\infty)$, $c_0^{\pm}>0$, so that 
in Regimes $\Rf_2$ on $(V'(u^*_{+}), \infty)$ and in $\Rf_{3b}, \Rf_{3b}$ on $(-\infty, V'(u_-^*))$,
$v_{1}^{\infty}$ is always a positive adjustment to the zero-order term $v_0^{\infty}$;
see Figure~\ref{fig:asymeffect} for an example of this 'convexity effect'.
\end{enumerate}
\end{remark}

Theorem~\ref{theorem:HestonLargeMatFwdSmileNonSteep} displays varying levels of degeneration for high-order forward smile asymptotics.
In $\Rf_1$ one can in principle obtain arbitrarily high-order asymptotics.
In $\Rf_2$, $\Rf_{3a}$ and $\Rf_{3b}$ one can only specify the forward smile to arbitrary order if $\mu=1/2$. 
If this is not the case then we can only specify the forward smile to first order. 
Now the dynamics of the Heston volatility $\sigma_t:=\sqrt{V_t}$ is given by
$
\D \sigma_t=\left(\frac{2\mu -1}{8\sigma_t}\xi^2-\frac{\kappa \sigma_t}{2}\right)\D t +\frac{\xi}{2} \D W_t,
$
with $\sigma_0=\sqrt{v}$.
If $\mu=1/2$ then the volatility becomes Gaussian, which this corresponds to a specific case of the Sch\"obel-Zhu stochastic volatility model. 
So as the Heston volatility dynamics deviate from Gaussian volatility dynamics a certain degeneracy occurs such that one cannot specify high order forward smile asymptotics.
Interestingly, a similar degeneracy occurs in~\cite{JR13} for exploding small-maturity Heston forward smile asymptotics and in~\cite{DFJV} when studying the tail probability of the stock price.
As proved in~\cite{DFJV}, the square-root behaviour of the variance process induces some singularity and hence
a fundamentally different behaviour when $\mu \ne 1/2$.
In $\Rf_2$, $\Rf_{3a}$ and $\Rf_{3b}$ at the boundary points $V'(u^*_{\pm})$ one cannot specify the forward smile 
beyond first order for any parameter configurations. 
This could be because these asymptotic regimes are extreme in the sense that they are transition points between standard and degenerate behaviours and therefore difficult to match with BSM forward volatility.
Finally in $\Rf_{3b}$ and $\Rf_4$ for $k>V'(1)$ we obtain the most extreme behaviour, in the sense that one cannot specify the forward smile beyond zeroth order.
This is however not that surprising since the large correlation regime has fundamentally different behaviour to the BSM model (see also Remark~\ref{rem:largematfwdstartoption}(iii)).


\subsection{SVI-type limits}\label{sec:SVIExtended}
The so-called 'Stochastic Volatility Inspired' (SVI) parametrisation of the spot implied volatility smile was proposed in~\cite{G04}. 
As proved in~\cite{GJ11}, under the assumption $\kappa>\rho\xi$, 
the SVI parametrisation turn out to be the true large-maturity limit for the Heston (spot) smile.
We now extend these results to the large-maturity forward implied volatility smile. 
Define the following extended SVI parametrisation
$$
\sigma^2_{\SVI}(k,a,b,r,m,s,i_0,i_1,i_2):=a+b\left(r(k-m)+i_0 \sqrt{i_1(k-m)^2+i_2(k-m)+i_0 s^2}\right),
$$
for all $k\in\RR$ and the constants
\begin{equation*}
\left\{
\begin{array}{rl}
\omega_1 & := \displaystyle\frac{2\mu}{1-\rho^2}
\left(\sqrt{\left(2\kappa+\xi^2-\rho\xi\right)^2+\xi^2\left(1-\rho^2\right)}-\left(2\kappa+\xi^2-\rho\xi\right)\right),
\qquad
\omega_2:=\frac{\xi}{\kappa\theta},\\ 
a_{\pm} & := \displaystyle\frac{\kappa\theta }{2 \left(u^*_{\pm}-1\right) u^*_{\pm} \beta _t},
\quad  
b_{\pm} := 4 \sqrt{\left(u^*_{\pm}-1\right) u^*_{\pm}},
\quad 
r_{\pm} := \frac{2(2 u^*_{\pm}-1)}{b_{\pm}}, 
\quad m_{\pm}:=\left(u^*_{\pm}-\frac{1}{2}\right)a_{\pm},\\
\widetilde{a} & := -2 \widetilde{m},
\quad  
\widetilde{b} := 4 \sqrt{-\widetilde{m}},
\quad 
\widetilde{r} := \frac{1}{2\sqrt{-\widetilde{m}}}, 
\quad 
\widetilde{m}:=\mu(\kappa-\rho\xi),
\end{array}
\right.
\end{equation*}
where  $u^*_{\pm}$  is defined in~\eqref{eq:DefUpmU*pm} and $\beta_t$ in~\eqref{eq:DGammaBeta}.
Define the following combinations:
\begin{equation}
\begin{array}{lllllllll}
\Ss_0: 
 & a=\frac{\omega_1(1-\rho)^2}{2}, & b=\frac{\omega_1\omega_2}{2}, & r=\rho, & m=-\frac{\rho}{\omega_2}, & s=\frac{\sqrt{1-\rho^2}}{\omega_2}, & i_0=1,  & i_1=1,  & i_2=0,
 \\
\Ss_\pm: 
  & a=a_{\pm}, & b=b_{\pm}, & r=r_{\pm}, & m=m_{\pm}, & s=\frac{1}{8}a_{\pm}, & i_0=-1,  & i_1=1,  & i_2=0,
\\
\Ss_1: 
 & a=\widetilde{a}, & b=\widetilde{b}, & r=\widetilde{r}, & m=\widetilde{m}, & s=0, & i_0=1,  & i_1=0,  & i_2=1.
\end{array}
\end{equation}
The proof of the following result follows from simple manipulations of the zeroth-order forward smile in Theorem~\ref{theorem:HestonLargeMatFwdSmileNonSteep} using the characterisation of $V^*$ in
Lemma~\ref{lemma:V*Characterisation}.
\begin{corollary}\label{cor:SVINonSteep}
The pointwise continuous limit $\lim_{\tau\uparrow\infty}\sigma^2_{t,\tau}(k\tau)=\sigma^2_{\SVI}(k,a,b,r,m,s,i_0,i_1,i_2)$ exists for $k\in\mathbb{R}$ with constants $a,b,r,m,s,i_0,i_1$ and $i_2$ 
given by\footnote{whenever $\Ss_0$ is in force, the case $k=V'(a)$ 
is excluded if $v = \theta \Upsilon(a)$, with $\Upsilon$ defined in~\eqref{eq:Upsilon}, for $a\in\{0,1\}$.}:
\begin{itemize}
\item $\Rf_1$: 
$\Ss_0$ for $k\in\RR$;

\item $\Rf_2$:
$\Ss_0$ for $k\in(-\infty, V'(u_+^*))$;
$\Ss_+$ for $k\in[V'(u_+^*),+\infty)$;

\item $\Rf_{3a}$:
$\Ss_-$ for $k\in(-\infty, V'(u_-^*)]$;
$\Ss_0$ for $k\in(V'(u_-^*),+\infty)$;

\item $\Rf_{3b}$:
$\Ss_-$ for $k\in(-\infty, V'(u_-^*)]$;
$\Ss_0$ for $k\in(V'(u_-^*),V'(1))$;
$\Ss_1$ for $k\in[V'(1),+\infty)$;

\item $\Rf_4$:
$\Ss_0$ for $k\in(-\infty, V'(1))$;
$\Ss_1$ for $k\in[V'(1),+\infty)$.
\end{itemize}
\end{corollary}

\section{Financial Interpretation of the large-maturity regimes}\label{sec:finintuition}

The large-maturity regimes in~\eqref{eq:Regimes} were identified with specific properties of the limiting forward logarithmic moment generating function.
Each regime uncovers fundamental properties of the large-maturity forward smile, 
some of which having been empirically observed by practitioners.
These regimes are not merely mathematical curiosities, but their studies reveal particular behaviours and oddities
of the model.
An intuitive question is how different the large-maturity forward smile and the large-maturity spot smile are.
This is a metric that a trader would have a view on and can be analysed using historical data.
Because of the ergodic properties of the variance process, at first sight it seems natural 
to conjecture that the large-maturity spot and forward smiles should be the same at leading order.
More specifically, if $\sigma_{\tau}^{(t)}(k)$ denotes the Black-Scholes implied volatility observed at time $t$, 
i.e. the unique positive solution to the equation 
$\EE\left[(\E^{X_{t+\tau}-X_t}-\E^{k})^+|\mathcal{F}_t\right]=C_{\textrm{BS}}(\tau,k,\sigma_{\tau}^{(t)}(k))$,
then by definition the forward implied volatility solves the equation
$C_{\textrm{BS}}(\tau,k,\sigma_{t,\tau}(k))=\mathbb{E}[C_{\textrm{BS}}(\tau,k,\sigma_{\tau}^{(t)}(k))]$.
If we suppose that $\lim_{\tau\uparrow\infty}\sigma_{\tau}^{(t)}(\tau k)=\sigma^{\infty}(k)$, where the function $\sigma^{\infty}$ is independent of $t$ (this is the case in Heston --- it does not depend on $V_t$) then it seems reasonable to suppose that
$C_{\textrm{BS}}(\tau,k,\sigma_{t,\tau}(k\tau))\approx C_{\textrm{BS}}(\tau,k,\sigma^{\infty}(k))$
and hence that $\sigma_{t,\tau}(k\tau)\approx \sigma^{\infty}(k) \approx \sigma_{\tau}(k\tau)$.
It is therefore natural to conjecture (see for example~\cite{Ba}) that the limiting forward smile $\lim_{\tau\uparrow\infty}\sigma_{t,\tau}(k\tau)$
is the same as the limiting spot smile $\lim_{\tau\uparrow\infty}\sigma_{\tau}(k\tau)$.
Theorem~\ref{theorem:HestonLargeMatFwdSmileNonSteep} shows us that this only holds under the good correlation regime $\Rf_1$, i.e. for correlations `close' to zero. 
Deviations of the correlation from zero therefore effect how different the large-maturity forward smile is to the large-maturity spot smile.

Consider now the practically relevant (on Equity markets) case of large negative correlation ($\Rr_2$).
In Figure~\ref{fig:asymeffect} we compare the two limiting smiles using the zero-order asymptotics in Corollary~\ref{cor:SVINonSteep} when $\rho<\rho_{-}$. 
At the critical log-strike $V'(u^*_{+})$, the forward smile becomes more convex than the corresponding spot smile.
Interestingly this asymmetric feature has been empirically observed by practitioners~\cite{B04} and is a fundamental property 
of the model---not only for large-maturities.
Quoting Bergomi~\cite{B04} from an empirical analysis: "...the increased convexity (of the forward smile) with respect to today's smile is larger for $k>0$ than for $k<0$...this is specific to the Heston model."

\begin{figure}[h!tb] 
\centering
\includegraphics[scale=0.7]{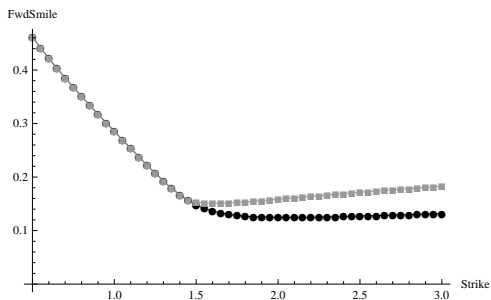}
\caption{
Here $t=0.5, \tau=2, v=\theta=0.1, \kappa=2, \xi=1, \rho=-0.9$, so that $\Rf_2$ applies.
Circles correspond to the spot smile $K\mapsto\sigma_{\tau}(\log K)$ 
and squares to the forward smile $K\mapsto\sigma_{t,\tau}(\log K)$
using the zeroth-order asymptotics in Corollary~\ref{cor:SVINonSteep}.
Here $\rho_{-}\approx-0.63$ and $\E^{2V'(u^*_{+})}\approx 1.41$.}
\label{fig:asymeffect}
\end{figure}

It is natural to wonder about the origin of this effect.
Consider a standard European option with large strike $k>0$. 
A large number of sample paths of the stock price approach $\E^{k}$, 
but, because of the negative correlation, the corresponding variance tends to be low
(the so-called `leverage effect').
For a delta-hedged long position this is exactly where we want the variance to be the highest 
(maximum gamma and vega).
Hence there is a tendency for the (spot) implied volatility to be downward sloping for high strikes.
On the other hand, consider a forward-start option with large strike $k>0$.
Suppose that the variance is large at the forward-start date, $t$.
Because of the negative correlation, the stock price will tend to be low here.
But this is irrelevant since the stock price is always re-normalised to $1$ at this point.
Hence there will be a greater number of paths where the re-normalised stock $S_u/S_{t}$ for $t\leq u\leq t+\tau$ is close to $\E^{k}$ and the variance is high relative to the (spot) case discussed above.
The relative nature of this effect induces this `convexity effect'.
When there is large positive correlation $\rho>\rho_{+}>0$ ($\Rr_3$), then the large-maturity forward smile is more convex then the large-maturity spot smile for \textit{low} strikes, $k<0$.
This is the `mirror image' effect of $\Rr_2$ and follows from similar intuition to above.

When $\rho>\kappa/\xi$ ($\Rr_{3b}$ and $\Rr_4$) there is a transition point for large strikes where the smile is upward sloping and possibly concave (See Figures~\ref{fig:LargeCorrelRegime10years} and~\ref{fig:LargeCorrelRegime20years}). 
It is important to note that this effect materialises for both the large-maturity spot and forward smile and is due to the fact that 
paths where the stock price are high will tend to be accompanied by periods of very high variance because of the positive correlation.

The intuitive arguments given above for each regime are not specific to Heston.
A natural conjecture is that all stochastic volatility models where the variance process has a stationary distribution will exhibit similar large-maturity regimes.
However, the location of the transition points and the magnitude of the `convexity corrections' may be quite different and model specific.

\section{Numerics}\label{sec:nonsteepnumerics}

We first compare the true Heston forward smile and the asymptotics developed in the paper.
We calculate forward-start option prices using the inverse Fourier transform representation 
in~\cite[Theorem 5.1]{L04F}
and a global adaptive Gauss-Kronrod quadrature scheme. 
We then compute the forward smile $\sigma_{t,\tau}$ with a simple root-finding algorithm.
In Figure~\ref{fig:HestLargeMatComp} we compare the true forward smile using Fourier inversion and the asymptotic in Theorem~\ref{theorem:HestonLargeMatFwdSmileNonSteep}(i) for the good correlation regime,
which was derived in~\cite{JR12}.
In Figure~\ref{fig:5yearnonsteep} we compare the true forward smile using Fourier inversion and the asymptotic in Theorem~\ref{theorem:HestonLargeMatFwdSmileNonSteep}(ii) for the asymmetric negative correlation regime.
Higher-order terms are computed using the theoretical results above; 
these can in principle be extended to higher order, but the formulae become rather cumbersome;
numerically, these higher-order computations seem to add little value to the accuracy anyway.
In Figure~\ref{fig:AsymNegTrans} we compare the asymptotic in Theorem~\ref{theorem:HestonLargeMatFwdSmileNonSteep}(ii) for the transition strike $k=V'(u^*_{+})$.
Results are all in line with expectations.

In the large correlation regime $\Rf_4$, we find it more accurate to use Theorem~\ref{theorem:largematasympcalls} and then numerically invert the price to get the corresponding forward smile (Figures~\ref{fig:LargeCorrelRegime10years} and~\ref{fig:LargeCorrelRegime20years}), 
rather than use the forward smile asymptotic in Theorem~\ref{theorem:HestonLargeMatFwdSmileNonSteep}.
As explained in Remark~\ref{rem:largematfwdstartoption}(iv) the leading-order accuracy of option prices in this regime is poor and higher-order terms embed important distinctions that need to be included.
This also explains the poor accuracy of the forward smile asymptotic in Theorem~\ref{theorem:HestonLargeMatFwdSmileNonSteep} for the large correlation regime.
As seen in the proof (Section~\ref{sec:fwdsmileproofnonsteeplargemat}), 
the leading-order behaviour of option prices is used to line up strike domains in the BSM and Heston model 
and then forward smile asymptotics are matched between the models.
If the leading-order behaviour is poor, then regardless of the order of the forward smile asymptotic, 
there will always be a mismatch between the asymptotic forms and the forward smile asymptotic will be poor.
Using the approach above bypasses this effect and is extremely accurate already at first order (Figures~\ref{fig:LargeCorrelRegime10years} and~\ref{fig:LargeCorrelRegime20years}).

In all but $\Rf_1$, higher-order terms can approach zero or infinity as the strike approaches the critical values ($V'(u^*_{+})$ or $V'(1)$), separating the asymptotic regimes, and forward smile 
(and forward-start option price) asymptotics are not continuous there (apart from the zeroth-order term), 
see also Remark~\ref{rem:contlargemat}(i).
Numerically this implies that the asymptotic formula may break down for strikes in a region around the the critical strike.
Similar features have been observed in~\cite{JR13} where degenerate asymptotics were derived for the exploding small-maturity Heston forward smile.

\begin{figure}
\centering
\mbox{\subfigure[Asymptotic vs Fourier Inversion.]{\includegraphics[scale=0.6]{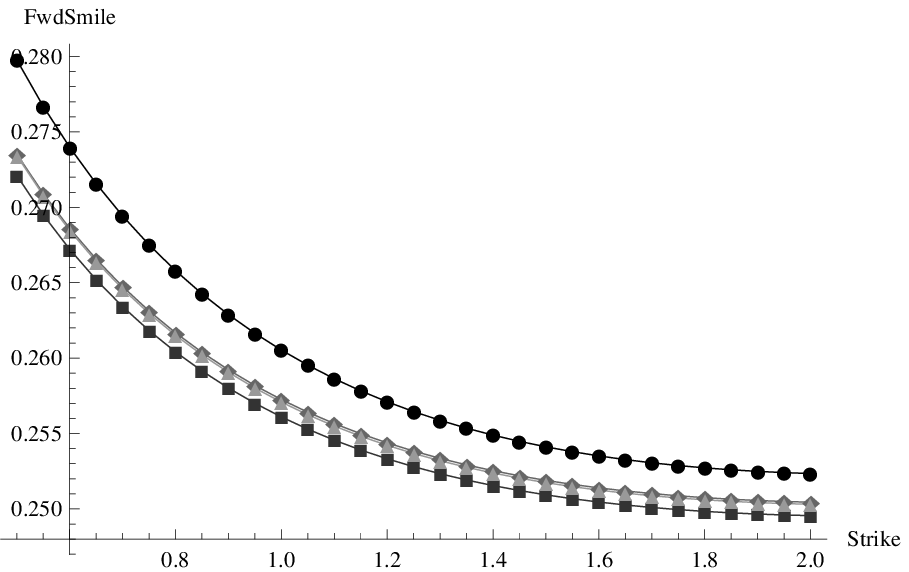}}\quad
\subfigure[Errors]{\includegraphics[scale=0.6]{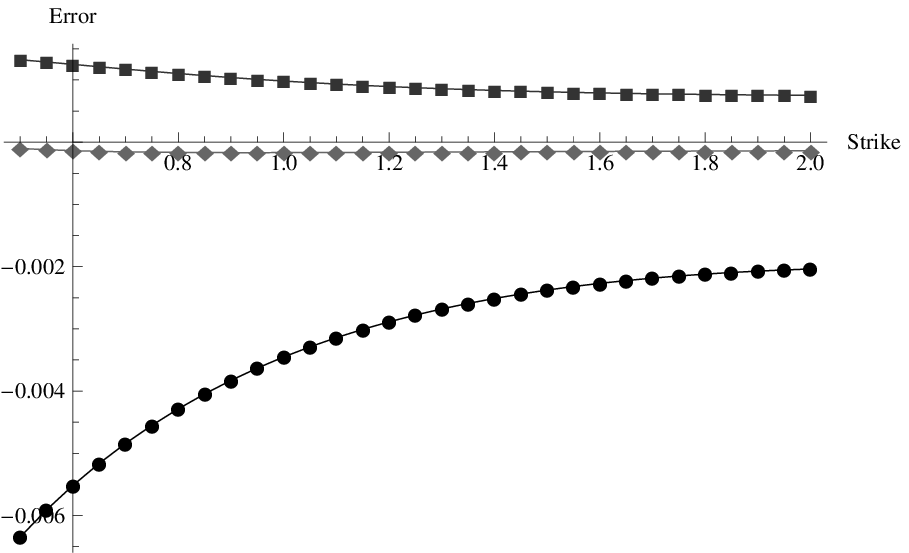}}}
\caption{\textbf{Good correlation regime $\Rf_1$.} 
In (a) circles, squares and diamonds represent the zeroth-, first-and second-order asymptotics respectively and triangles represent the true forward smile. 
In (b) we plot the differences between the true forward smile and the asymptotic. 
Here $t=1$, $\tau=5$ and $v=0.07$, $\theta=0.07$, $\kappa=1.5$, $\xi=0.34$, $\rho=-0.25$.}
\label{fig:HestLargeMatComp}
\end{figure}

\begin{figure}
\centering
\mbox{\subfigure[Asymptotic vs Fourier inversion.]{\includegraphics[scale=0.6]{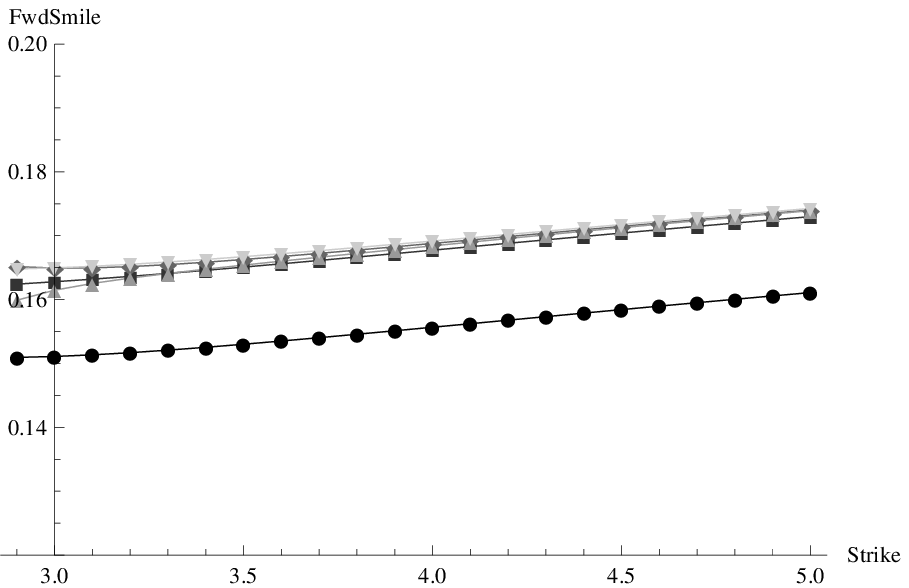}}\quad
\subfigure[Errors.]{\includegraphics[scale=0.6]{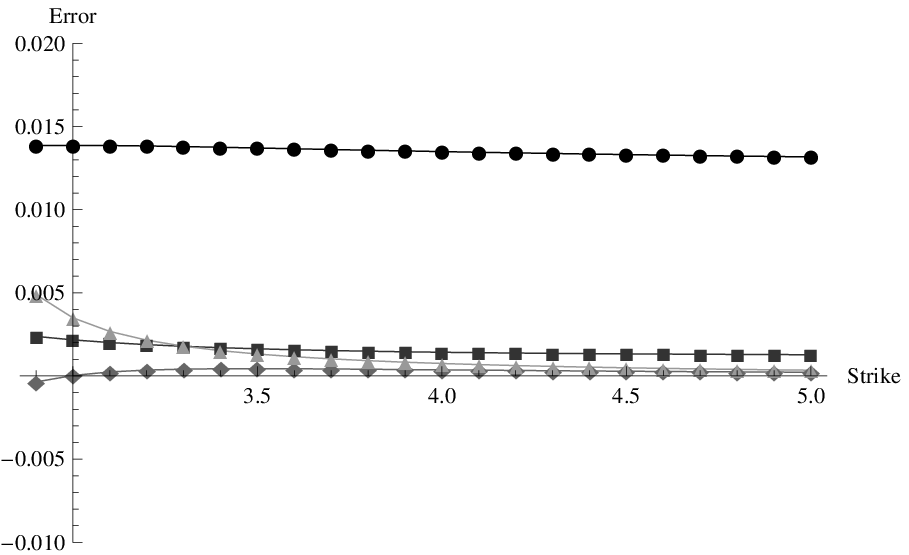}}}
\caption{\textbf{Asymmetric correlation regime~$\Rf_{2}$.} 
Here $t=1$, $\tau=5$ and $v=\theta=0.07$, $\rho=-0.8$, $\xi=0.65$ and $\kappa=1.5$, 
which implies $\E^{V'(u^*_{+})\tau}\approx 2.39$.
In (a) circles, squares, diamonds and triangles represent the zeroth-, first-, second- and third-order asymptotics respectively and backwards triangles represent the true forward smile. 
In (b) we plot the errors. }
\label{fig:5yearnonsteep}
\end{figure}

\begin{figure}
\centering
\mbox{\subfigure[Asymptotic vs Fourier inversion.]{\includegraphics[scale=0.6]{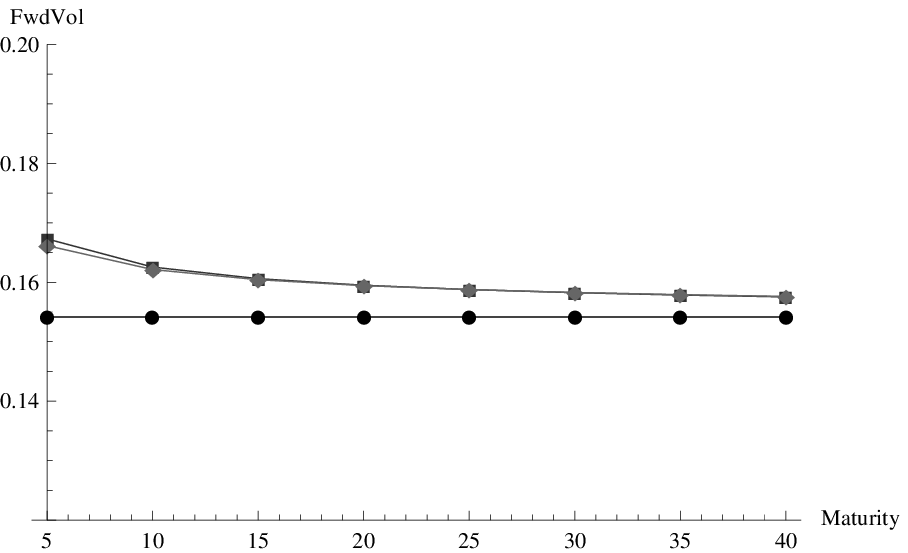}}\quad
\subfigure[Errors.]{\includegraphics[scale=0.6]{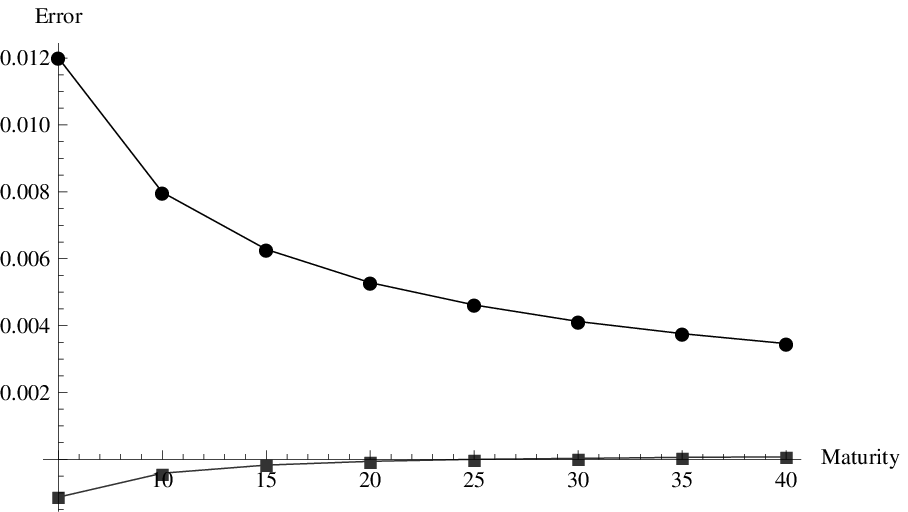}}}
\caption{\textbf{Asymmetric correlation regime~$\Rf_{2}$.} 
Here $t=1$ and the Heston parameters are the same as in Figure~\ref{fig:5yearnonsteep}. 
Circles and squares represent the zeroth- and first-order asymptotic and triangles represent the true forward smile.
The horizontal axis is the maturity and the strike is equal $\E^{V'(u^*_{+})\tau}$.
In (b) we plot the errors.}
\label{fig:AsymNegTrans}
\end{figure}

\begin{figure}
\centering
\mbox{\subfigure[Asymptotic vs Fourier inversion.]{\includegraphics[scale=0.6]{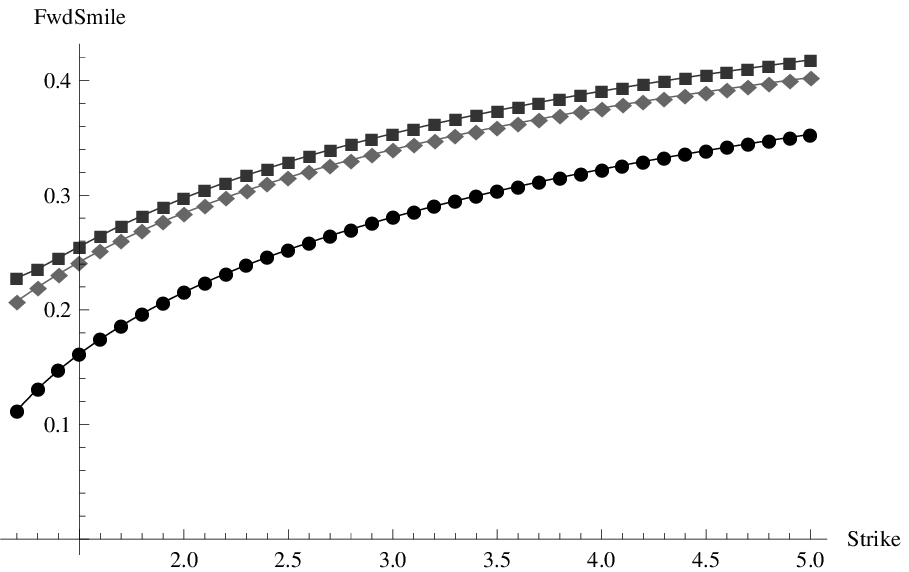}}\quad
\subfigure[Errors.]{\includegraphics[scale=0.6]{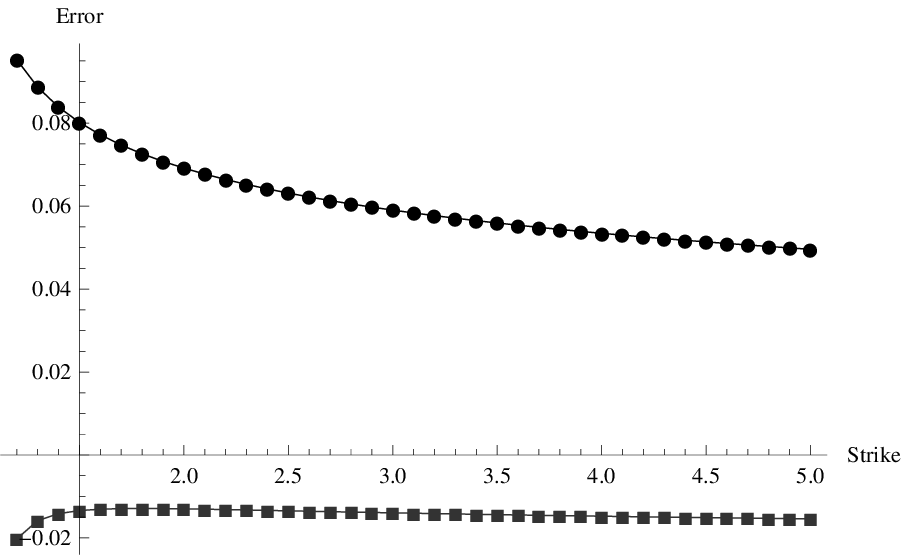}}}
\caption{\textbf{Large correlation regime~$\Rf_{4}$.} 
Here $t=0$, $\tau=10$, $v=\theta=0.07$, $\rho=0.5$, $\xi=0.6$, and $\kappa=0.1$. 
Circles and squares represent the zeroth- and first-order asymptotic and triangles represent the true forward smile.
Further $\E^{V'(1)\tau} \approx 1.06$.}
\label{fig:LargeCorrelRegime10years}
\end{figure}

\begin{figure}
\centering
\mbox{\subfigure[Asymptotic vs Fourier inversion.]{\includegraphics[scale=0.6]{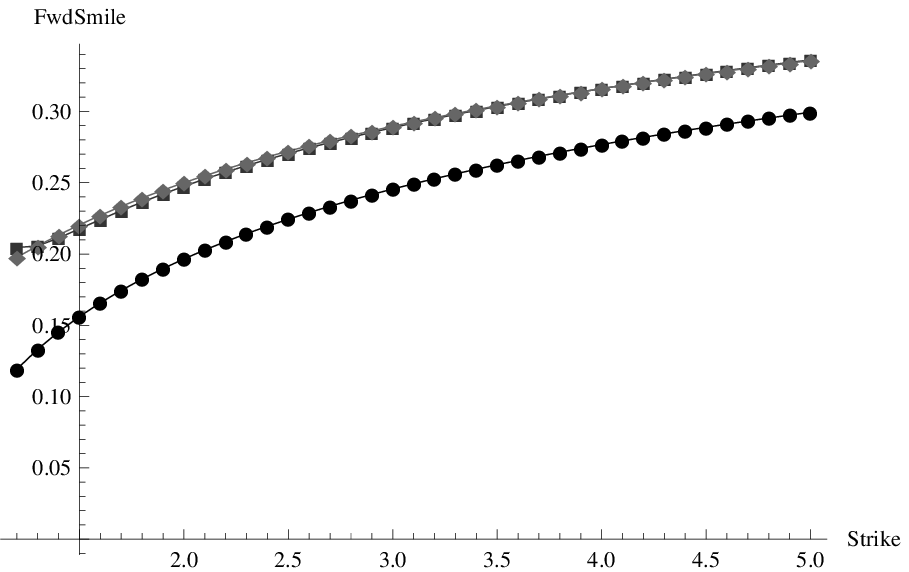}}\quad
\subfigure[Errors.]{\includegraphics[scale=0.6]{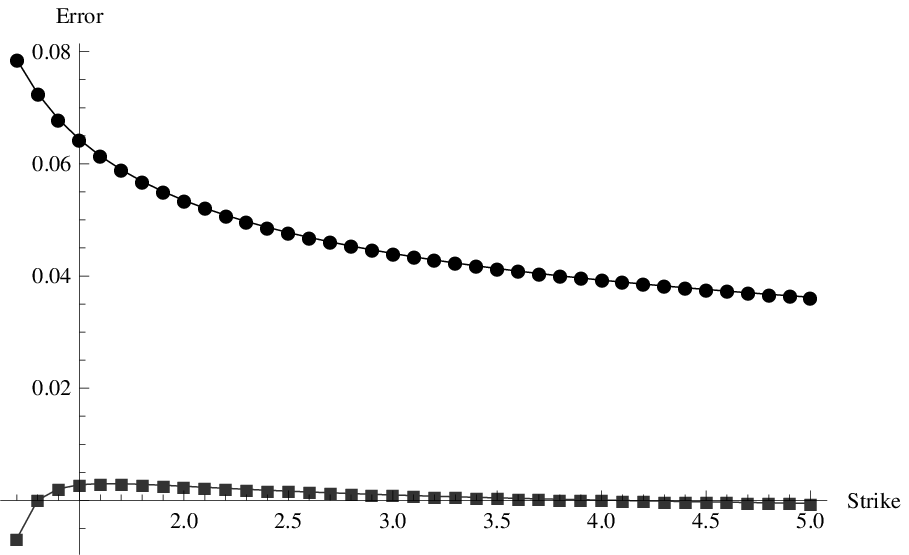}}}
\caption{\textbf{Large correlation regime~$\Rf_{4}$.}
Here $t=0$, $\tau=20$ and the Heston parameters are the same as in Figure~\ref{fig:LargeCorrelRegime10years}.
Circles and squares represent the zeroth- and first-order asymptotic and triangles represent the true forward smile.}
\label{fig:LargeCorrelRegime20years}
\end{figure}

\section{Proof of Theorems~\ref{theorem:largematasympcalls} and~\ref{theorem:HestonLargeMatFwdSmileNonSteep}} \label{sec:proofsnonsteep}
This section is devoted to the proofs of the option price and implied volatility expansions
in Theorems~\ref{theorem:largematasympcalls} and~\ref{theorem:HestonLargeMatFwdSmileNonSteep}.
We first start (Section~\ref{sec:FwdLimits}) with some preliminary results of the behaviour
of the moment generating function of the forward process $(X_\tau^{(t)})_{\tau>0}$,
on which the proofs will rely.
The remainder of the section is devoted to the different cases, as follows:
\begin{itemize}
\item Section~\ref{sec:ProofConvex} is the easy case, namely whenever the function~$V^*$ in~\eqref{eq:VStarDefinition} is strictly convex, corresponding to the behaviour $\Hh_0$, except at the points $V'(0)$ and $V'(1)$.
\item In Section~\ref{sec:Genmethlargetime}, we outline the general methodology we shall use in all other cases:
\begin{itemize}
\item Section~\ref{sec:asymmetricproofs} tackles the cases $\Hh_\pm$, $\widetilde{\Hh}_\pm$ and $\Hh_2$, 
corresponding to the function~$V^*$ being linear;
\item Section~\ref{sec:V0V1limFT} is devoted to the analysis at the points $V'(0)$ and $V'(1)$
\end{itemize}
\item Section~\ref{sec:fwdsmileproofnonsteeplargemat} translates the expansions for the option price
into expansions for the implied volatility.
\end{itemize}

\subsection{Forward logarithmic moment generating function (lmgf) expansion 
and limiting domain}\label{sec:FwdLimits}
For any $t\geq 0$, $\tau>0$, define the re-normalised lmgf of $X_{\tau}^{(t)}$ 
and its effective domain~$\mathcal{D}_{t,\tau}$
by 
\begin{equation}\label{eq:MGFFwd}
\Lambda^{(t)}_\tau\left(u\right):=\tau^{-1}\log\EE\left(\E^{uX_{\tau}^{(t)}}\right),
\quad\text{for all }
u\in\mathcal{D}_{t,\tau}:=\{u\in\RR:|\Lambda^{(t)}_\tau\left(u\right)|<\infty\}.
\end{equation}
A straightforward application of the tower property for expectations 
yields:
\begin{equation}\label{eq:LambdaTau}
\tau\Lambda^{(t)}_\tau\left(u\right)
=
A\left(u,\tau\right)+\frac{B(u,\tau)v\E^{-\kappa t}}{1-2\beta_t B(u,\tau)}
-\mu\log\left(1-2\beta_t B\left(u,\tau\right)\right),
\quad\text{for all }u\in\mathcal{D}_{t,\tau},
\end{equation}
where 
\begin{align}
A(u,\tau) & := 
\frac{\mu}{2}\left(\left(\kappa-\rho\xi u- d\left(u\right)\right)\tau-2\log\left(\frac{1-\gamma\left(u\right)\exp\left(-d\left(u\right)\tau\right)}{1-\gamma\left(u\right)}\right)\right),\nonumber\\
B(u,\tau) & := \frac{\kappa-\rho\xi u-d(u)}{\xi^2}\frac{1-\exp\left(-d\left(u\right)\tau\right)}{1-\gamma\left(u\right)\exp\left(-d\left(u\right)\tau\right)},\nonumber\\
d(u) & := \left(\left(\kappa-\rho\xi u\right)^2+u\left(1-u\right)\xi^2\right)^{1/2},
\qquad
\gamma(u) := \frac{\kappa-\rho\xi u-d\left(u\right)}{\kappa-\rho\xi u+d\left(u\right)},
\qquad
\beta_t  := \frac{\xi^2}{4\kappa}\left(1-\E^{-\kappa t}\right).
\label{eq:DGammaBeta}
\end{align}
The first step is to characterise the effective domain $\mathcal{D}_{t,\tau}$ for fixed $t\geq0$
as $\tau$ tends to infinity.
Recall that the large-maturity regimes are defined in~\eqref{eq:Regimes} with $u_{\pm}$ and $u^*_{\pm}$ given in ~\eqref{eq:DefUpmU*pm}.
\begin{lemma}
For fixed $t\geq 0$, $\mathcal{D}_{t,\tau}$ converges (in the set sense) to $\mathcal{D}_{\infty}$
defined in Table~\ref{eq:DInfinityLargeMaturity}, as $\tau$ tends to infinity.
\end{lemma}

\begin{proof}
Recall the following facts from~\cite[Lemma 5.11 and Proposition 5.12]{JR12} and ~\cite[Proposition 2.3]{JM12}, 
with the convention that $u_{\pm}^{*}=\pm\infty$ when $t=0$:
\begin{enumerate}[(i)]
\item $ [0,1] \subset [u_{-},u_{+}] \cap(-\infty,u^*_{+})\subset\mathcal{D}_{t,\tau}$ for all $\tau>0$ if $\rho<0$; 
\item $ [0,1] \subset [u_{-},u_{+}] \cap(u^*_{-},\infty)\subset\mathcal{D}_{t,\tau}$ for all $\tau>0$ if $0<\rho\leq \kappa/\xi$;
\item $ [0,1] \subset [u_{-},u_{+}] \subset\mathcal{D}_{t,\tau}$ for all $\tau>0$ if $\rho=0$;
\item $ [0,1] \subset [u_{-},1] \cap (u^*_{-},\infty) \subset\mathcal{D}_{t,\tau}$ for all $\tau>0$ if $\rho> \kappa/\xi$;
\item $1<u^*_{+}<u_{+}$ if and only if $\rho\in (-1,\rho_{-})$ and $u_{-}<u_{-}^*<0$ if and only if $\rho\in (\rho_{+},1)$. 
We always have $\rho_{-}\in(-1,0)$ and $\rho_{+}>1/2$. 
In the latter case it is possible that $\rho_{+}\geq1$ in which case $u_{-}^*\leq u_{-}$.
\end{enumerate}
Then for fixed $t\geq0$, the lemma follows directly from (i)-(iv) in combination with property (v).
\end{proof}
The following lemma provides the asymptotic behaviour of $\Lambda_{\tau}^{(t)}$ as $\tau$ tends to infinity.
The proof follows the same steps as~\cite[Lemma 5.13]{JR12}, using the fact that the asset price process $(\E^{X_t})_{t>0}$ is a true martingale~\cite[Proposition 2.5]{AP07}, and is therefore omitted.
\begin{lemma}\label{lem:fwdmgflargetauexpansion}
The following expansion holds for the forward lmgf $\Lambda^{(t)}_{\tau}$ defined in~\eqref{eq:MGFFwd}
($V$ and $H$ given in~\eqref{eq:VandH}):
\begin{equation*}
\Lambda^{(t)}_{\tau}(u)=
\left\{
\begin{array}{ll}
\displaystyle V(u)+\tau^{-1}H(u)\left(1+\mathcal{O}\left(\E^{-d\left(u\right)\tau}\right)\right),
&\quad\text{for all }u\in\mathcal{D}_{\infty}^{o}\setminus\{1\}, \text{ as }\tau\text{ tends to infinity},\\
\displaystyle 0, &\quad\text{for }u=1\text{ and all } \tau>0.
\end{array}
\right.
\end{equation*}
\end{lemma}
\begin{remark} \ \label{rem:largematnonsteep}
\begin{enumerate}[(i)] 
\item When $\rho>\kappa/\xi$ ($\Rf_{3b}$ and $\Rf_4$), we have $\lim_{u \uparrow 1}\Lambda^{(t)}_{\tau}(u)=V(1)\neq 0$, 
so that the limit is not continuous at the right boundary $u=1$.
For $\rho\leq\kappa/\xi$ we always have $V(1)=H(1)=0$ and $1\in\mathcal{D}_{\infty}^{o}$.
\item For all $u\in\mathcal{D}_{\infty}^{o}$, $d(u)>0$, so that the remainder goes to zero exponentially fast as $\tau$ tends to infinity. 
\end{enumerate}
\end{remark}

\subsection{The strictly convex case}\label{sec:ProofConvex}
Let $\overline{k}:=\sup_{a\in\mathcal{D}_{\infty}}V'(a)$ and $\underline{k}:=\inf_{a\in\mathcal{D}_{\infty}}V'(a)$. 
When $k\in(\underline{k},\overline{k})\setminus\{V'(0), V'(1)\}$, 
an analogous analysis to~\cite[Theorem 2.4, Propositions 2.12 and 3.5]{JR12}, essentially based on the strict convexity of~$V$ on $(\underline{k}, \overline{k})$, 
can be carried out and we immediately obtain the following results for forward-start option prices and forward implied volatilities 
(hence proving Theorems~\ref{theorem:largematasympcalls} and~\ref{theorem:HestonLargeMatFwdSmileNonSteep} when $\Hh_0$ holds):
\begin{lemma}\label{lem:fwdstoptsteep}
The following expansions hold for all $k\in(\underline{k},\overline{k})\setminus\{V'(0), V'(1)\}$ as $\tau$ tends to infinity:
\begin{align*}
\EE\left(\E^{X^{(t)}_{\tau}}-\E^{k\tau}\right)^+
&=
\Ii\left(k,\tau,V'(0),V'(1),0\right) + 
\frac{\phi_0(k,t)}{\tau^{1/2}}\E^{-\tau\left(V ^*(k)-k \right)}
\left(1+\mathcal{O}\left(\tau^{-1}\right)\right), \\
\sigma_{t,\tau}^2(k\tau)&=v_0^{\infty}(k,t)+\frac{8 v_0^{\infty}(k,t)^2}{4 k^2-v_0^{\infty}(k,t)^2}\chi_0(k,t)\tau^{-1}+\mathcal{O}(\tau^{-2}),
\end{align*}
with $V^*$ given in Lemma~\ref{lemma:V*Characterisation}, $\Ii$ and $\phi_0$ in~\eqref{eq:intrinsiclargetimenonsteep} and~\eqref{eq:phi0largetime}, $v_0^{\infty}$ in~\eqref{eq:v0nonsteep}, $\chi_0$ in~\eqref{eq:chi0largetime} and 
\begin{equation}\label{eq:steepstrikes}
(\underline{k},\overline{k}) = 
\left\{
\begin{array}{ll}
\RR, & \text{in }\Rf_1,\\
(-\infty,V'(u^*_{+})), & \text{in }\Rf_2,\\
(V'(u^*_{-}),+\infty), & \text{in }\Rf_{3a},\\
(V'(u^*_{-}),V'(1)), & \text{in }\Rf_{3b},\\
(-\infty,V'(1)), & \text{in }\Rf_4.
\end{array}
\right.
\end{equation}
\end{lemma}
\begin{proof}
We sketch here a quick outline of the proof.
For any $k\in(\underline{k},\overline{k})$, the equation $V'(u^*(k))=k$ has a unique solution $u^*(k)$
by strict convexity arguments.
Define the random variable $Z_{k,\tau}:=(X_{\tau}^{(t)} - k\tau)/\sqrt{\tau}$; using Fourier transform methods analogous to~\cite[Theorem 2.4, Proposition 2.12]{JR12}) the option price reads, for large enough~$\tau$,
$$
\EE\left[\E^{X^{(t)}_{\tau}}-\E^{k\tau}\right]^+
=
\Ii\left(k,\tau,V'(0),V'(1),0\right) +
\frac{\E^{-\tau(k(u^*(k)-1)-V(u^*(k)))}\E^{H(u^*(k))}}{2\pi}
\int_{\RR}\frac{\Phi_{\tau,k}(u)\sqrt{\tau}\D u}{[u-\I\sqrt{\tau}(u^*(k)-1)][u-\I\sqrt{\tau}u^*(k)]},
$$
where $\Phi_{\tau,k}(u)\equiv\EE^{\widetilde{\QQ}_{k,\tau}}(\E^{\I u Z_{k,\tau}})$ 
is the characteristic function of~$Z_{k,\tau}$ under the new measure~$\widetilde{\QQ}_{k,\tau}$ defined by
$\frac{\D\widetilde{\QQ}_{k,\tau}}{\D\PP}:=\exp\left( u^*(k)X_{\tau}^{(t)}-\tau \Lambda_{\tau}^{(t)}(u^*(k)\right)$.
Using Lemma~\ref{lem:fwdmgflargetauexpansion}, the proofs of the option price
and the forward smile expansions are similar to those of~\cite[Theorem 2.4 and Proposition 2.12]{JR12}
and~\cite[Proposition 3.5]{JR12}.
The exact representation of the set $(\underline{k},\overline{k})$ follows from 
the definition of $\mathcal{D}_{\infty}$ in Table~\ref{eq:DInfinityLargeMaturity} and the properties of~$V$.
\end{proof}

\subsection{Other cases: general methodology}\label{sec:Genmethlargetime}
Suppose that $\overline{k}$ (defined in Section~\ref{sec:ProofConvex}) is finite with $V'(\overline{u})=\overline{k}$.
We cannot define a change of measure (as in the proof of Lemma~\ref{lem:fwdstoptsteep}) by simply replacing $u^*(k)\equiv\overline{u}$ for $k\geq\overline{k}$ since the forward  lmgf
$\Lambda_{\tau}^{(t)}$ explodes at these points as $\tau$ tends to infinity (see Figure~\ref{fig:mgfexplosion}).
\begin{figure}[h!tb] 
\centering
\includegraphics[scale=0.7]{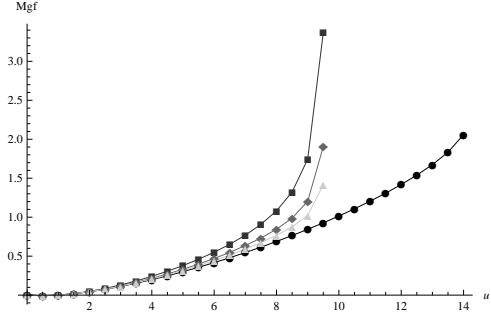}
\caption{
Regime $\Rf_2$: Circles  plot $u\mapsto V(u)$. 
Squares, diamonds and triangles  plot $u\mapsto V(u)+H(u)/\tau$ with $t=1$ and $\tau=2,5,10$. 
Heston model parameters are $v=0.07$, $\theta=0.07$, $\rho=-0.8$, $\xi=0.65$ and $\kappa=1.5$. 
Also $\rho_{-}\approx -0.56$, $u^*_{+}\approx 9.72$ and $u_{+}\approx 14.12$. 
}
\label{fig:mgfexplosion}
\end{figure}
One of the objectives of the analysis is to understand the explosion rate of the forward lmgf at these boundary points.
The key observation is that just before infinity, the forward lmgf 
$\Lambda_{\tau}^{(t)}$ is still steep on $\mathcal{D}_{t,\tau}^{o}$, 
and an analogous measure change to the one above can be constructed.
We therefore introduce the time-dependent change of measure
\begin{equation}\label{eq:MeasureChange}
\frac{\D\QQ_{k,\tau}}{\D\PP}:=\exp\left( u^*_{\tau}(k)X_{\tau}^{(t)}-\tau\Lambda^{(t)}_{\tau}(u^*_{\tau}(k))\right),
\end{equation}
where $u^*_{\tau}(k)$ is the unique solution to the equation $\partial_{u}\Lambda^{(t)}_{\tau}(u_{\tau}^*(k))=k$ for $k\geq \overline{k}$.
We shall also require that there exists $\tau_1>0$ such that $u^*_{\tau}(k)\in\mathcal{D}_{\infty}^{o}$ 
for all $\tau>\tau_1$ and $u^*_{\tau}\uparrow\overline{u}$;
therefore Lemma~\ref{lem:fwdmgflargetauexpansion} holds, 
and we can ignore the exponential remainder ($d(u)>0$ for all $u\in\mathcal{D}_{\infty}^{o}$) so that the equation $\partial_{u}\Lambda^{(t)}_{\tau}(u_{\tau}^*(k))=k$ reduces to \footnote{A similar analysis can be conducted even if $u^*_{\tau}(k)$ is not eventually in the interior of the limiting domain, but then one will need to use the full lmgf (not just the expansion) in~\eqref{eq:u^*tau}. }
\begin{equation}\label{eq:u^*tau}
V'\left(u^*_\tau(k)\right)+\tau^{-1}H'\left(u^*_\tau(k)\right)=k.
\end{equation}
In the analysis below, we will also require $u^*_{\tau}(k)$ to solve~\eqref{eq:u^*tau} 
and to converge to other points in the domain (not only boundary points).
This will be required to derive asymptotics under~$\Hh_0$ for the strikes~$V'(0)$ and~$V'(1)$, where there are no moment explosion issues but rather issues with the non-existence of the limiting Fourier transform (see Section~\ref{sec:V0V1limFT} for details).
We therefore make the following assumption:

\begin{assumption}\label{assump:ustartaueq}
There exists $\tau_1>0$ and a set $\mathcal{A}\subseteq\RR$ such that for all $\tau>\tau_1$ and 
$k\in\mathcal{A}$, 
Equation~\eqref{eq:u^*tau} admits a unique solution $u^*_\tau (k)$ on $\mathcal{D}_{\infty}^{o}$ 
satisfying
$\lim_{\tau\uparrow\infty}u^*_\tau(k) =  u^*_{\infty}\in\overline{\mathcal{D}_{\infty}}\cap (u_-,u_+)$.
\end{assumption}
Under this assumption
 $|\Lambda^{(t)}_{\tau}(u^*_\tau(k))|$ is finite for $\tau>\tau_1$ and $\mathcal{D}_{\infty}=\lim_{\tau\uparrow \infty} \{u\in\RR:|\Lambda_{\tau}^{(t)}(u)|<\infty\}$. 
Also $\D\QQ_{k,\tau}/\D\PP$ is almost surely strictly positive and by definition $\EE[\D\QQ_{k,\tau}/\D\PP]=1$. 
Therefore~\eqref{eq:MeasureChange} is a valid measure change for sufficiently large $\tau$ and all $k\in\mathcal{A}$.

Our next objective is to prove weak convergence of a rescaled version of the forward price process $(X_{\tau}^{(t)})_{\tau>0}$ under this new measure. 
To this end define the random variable 
$Z_{\tau,k,\alpha}:=(X_{\tau}^{(t)} - k\tau)/\tau^{\alpha}\label{eq:ztaukalpha}$
for $k\in\mathcal{A}$ and some $\alpha>0$,
with characteristic function 
$\Phi_{\tau,k,\alpha}:\RR\to\mathbb{C}$ under $\QQ_{k,\tau}$:
\begin{equation}\label{eq:CharacNonSteep}
\Phi_{\tau,k,\alpha}(u):=\EE^{\QQ_{k,\tau}}\left(\E^{\I u Z_{\tau,k,\alpha}}\right).
\end{equation}
Define now the functions $D:\RR_{+}^*\times\mathcal{A}\to\RR$ and $F:\RR^*_{+}\times \mathcal{A}\times\RR^*_{+}\to\RR$ by
\begin{equation}\label{eq:DF}
D(\tau,k):=\exp\Big[{-\tau\Big(k (u^*_{\tau}(k)-1)-V(u^*_{\tau}(k))\Big)}+H(u^*_{\tau}(k))\Big], \quad
F(\tau,k,\alpha):= \frac{1}{2\pi}\int_{\RR} \Phi_{\tau,k,\alpha}(u)\overline{C_{\tau,k,\alpha}(u)} \D u,
\end{equation}
where $\overline{C_{\tau,k,\alpha}(u)}$ denotes the complex conjugate of $C_{\tau,k,\alpha}$ in~\eqref{eq:Cdeflargetime}, namely:
\begin{equation}\label{eq:CConj}
\overline{C_{\tau,k,\alpha}(u)} =\frac{\tau^{\alpha}}{(u-\I \tau^{\alpha}(u^*_{\tau}-1)(u-\I \tau^{\alpha}u^*_{\tau})}.
\end{equation}
The main result here (proved in Appendix~\ref{sec:prooflargematfourtransf}) 
is an asymptotic representation for forward-start option prices:
\begin{lemma}\label{lem:nonsteeplem}
Under Assumption~\ref{assump:ustartaueq}, there exists $\beta>0$ such that for all $k\in\mathcal{A}$, 
as $\tau\uparrow\infty$:
\begin{align}
\EE\left(\E^{X_{\tau}^{(t)}}-\E^{k\tau}\right)^{+}
= \begin{dcases*}
       D(\tau,k) F(\tau,k,\alpha)\left(1+\mathcal{O}(\E^{-\beta\tau})\right),  & if  $u_{\tau}^*(k)>1$,\\
       (1-\E^{k\tau})+ D(\tau,k) F(\tau,k,\alpha)\left(1+\mathcal{O}(\E^{-\beta\tau})\right),  & if  $u_{\tau}^*(k)<0$, \\
  1+D(\tau,k) F(\tau,k,\alpha) \left(1+\mathcal{O}(\E^{-\beta\tau})\right),  & if $0<u_{\tau}^*(k)<1. $
        \end{dcases*}
\end{align}
\end{lemma}
We shall also need the following result on the behaviour of the characteristic function of~$Z_{\tau, k, \alpha}$
\begin{lemma}\label{lem:charactsimp}
Under Assumption~\ref{assump:ustartaueq} there exists $\beta>0$ such that for any $k\in\mathcal{A}$ as $\tau\uparrow\infty$:
$$
\Phi_{\tau,k,\alpha}(u) = \exp\Big(
-\I u k \tau^{1-\alpha} + 
\tau\left(V\left(\I u \tau^{-\alpha} + u^*_{\tau}\right) - V\left(u^*_{\tau}\right)\right)
 + H\left(\I u \tau^{-\alpha} + u^*_{\tau}\right) - H\left( u^*_{\tau}\right)
 \Big)
\left(1+\mathcal{O}(\E^{-\beta\tau})\right),
$$
where the remainder is uniform in $u$.
\end{lemma}
\begin{proof}
Fix $k\in\mathcal{A}$.
Analogous arguments to Lemma~\ref{lem:largetimesaddlefwdmgf}(iii) yield that 
$\Re d\left(\I u \tau^{-\alpha}+a\right)>d(a)$ for any $a\in\mathcal{D}_{\infty}^o$.
Assumption~\ref{assump:ustartaueq} implies that for all $\tau>\tau_1$, 
$\Re d\left(\I u \tau^{-\alpha}+u_{\tau}^*(k)\right)>d(u_{\tau}^*(k))$. 
It also implies that $u^*_\infty < u_+$, and hence there exists $\delta>0$ and $\tau_2>0$ such that $u_{\tau}^*(k)<u_{+}-\delta$
for all $\tau>\tau_2$.
Now, since $d$ is strictly positive and concave on $(u_{-},u_{+})$ and $d(u_{-})=d(u_{+})=0$,
we obtain $d(u_{\tau}^*(k))>d(u_{+}-\delta)>0$.
This implies that the quantities 
$\mathcal{O}\left(\exp\left[-d\left(\frac{\I u}{\tau^{\alpha}}+u_{\tau}^*(k)\right))\tau\right]\right)$ 
and 
$\mathcal{O}\left(\E^{-d(u_{\tau}^*(k))\tau}\right)$
are all equal to 
$\mathcal{O}\left(\E^{-d(u_{+}-\delta)\tau}\right)$ for all $k\in\mathcal{A}$.
Using the definition of $Z_{\tau,k,\alpha}$, the change of measure~\eqref{eq:MeasureChange} 
and Lemmas~\ref{lem:fwdmgflargetauexpansion} and~\ref{lem:largetimesaddlefwdmgf}, we can write
\begin{align*}
\log\Phi_{\tau,k,\alpha}(u)
&=\log \EE^{\QQ_{k,\tau}}\left[\E^{\I u Z_{\tau,k,\alpha}}\right]
=\log \EE\left[
\exp\left(u^*_{\tau}X_{\tau}-\tau\Lambda_{\tau}^{(t)}\left(u^*_{\tau}\right)
 + \frac{\I u}{\tau^{\alpha}}\left(X_{\tau} - k\tau\right)\right)\right] \\
&=-\I u k \tau^{1-\alpha} +\tau\left(\Lambda_{\tau}^{(t)}\left(\I u /\tau^{\alpha} + u^*_{\tau}\right)-\Lambda_{\tau}^{(t)}\left( u^*_{\tau}\right)\right) \\
&=-\frac{\I u k}{\tau^{\alpha-1}}
 +\tau\left[V\left(\frac{\I u}{\tau^{\alpha}} + u^*_{\tau}\right) - V(u^*_{\tau})\right]
 + H\left(\frac{\I u}{\tau^{\alpha}} + u^*_{\tau}\right)-H\left( u^*_{\tau}\right)
 + \mathcal{O}\left[\E^{-d\left(\I u \tau^{-\alpha}+u_{\tau}^*\right)\tau}\right]
 - \mathcal{O}\left(\E^{-d(u_{\tau}^*)}\tau\right)\\
&= -\I u k \tau^{1-\alpha} +\tau\left(V\left(\I u/ \tau^{\alpha} + u^*_{\tau}\right)-V\left( u^*_{\tau}\right)\right)+H\left(\I u /\tau^{\alpha} + u^*_{\tau}\right)-H\left( u^*_{\tau}\right)
+\mathcal{O}\left(\E^{-d(u_{+}-\delta)\tau}\right).
\end{align*}
Since $d(u_{+}-\delta)>0$ the remainder tends to zero exponentially fast as $\tau$ tends to infinity.
The uniformity of the remainder follows from tedious, yet non-technical, computations 
showing that the absolute value of the difference between $\log\Phi_{\tau,k,\alpha}(u)$ 
and its approximation is bounded by a constant independent of $u$ as $\tau$ tends to infinity.
\end{proof}

\subsection{Asymptotics in the case of extreme limiting moment explosions}\label{sec:asymmetricproofs}

We consider now the cases $\Hh_\pm$, $\widetilde{\Hh}_\pm$ and $\Hh_2$, 
corresponding to the limiting lmgf~$V$ being linear.
\begin{lemma} \label{lemma:u^*tau} 
Assumption~\ref{assump:ustartaueq} is verified in the following cases:
\begin{enumerate}[(i)]
\item
$\Rf_2$ with $\mathcal{A}=[V'(u^*_{+}),\infty)$ and $u_{\infty}^*=u_{+}^*$;
\item
$\Rf_{3a}$ and $\Rf_{3b}$ with $\mathcal{A}=(-\infty,V'(u^*_{-})]$ and $u_{\infty}^*=u_{-}^*$.
\item
$\Rf_{3b}$ and $\Rf_4$ with $\mathcal{A}=(V'(1),\infty]$ and $u_{\infty}^*=1$.
\end{enumerate}
\end{lemma}


\begin{proof}
Consider Case~(i) and re-write~\eqref{eq:u^*tau} as $H'(u^*_{\tau}(k))/\tau=k-V'(u^*_{\tau}(k))$. 
Let $k\geq V'(u_{+}^*)$;
since $V$ is strictly convex on $(u_{-},u_{+})$,
we have  $H'(u^*_{\tau}(k))/\tau=k-V'(u^*_{\tau}(k))\geq V'(u_{+}^*)-V'(u^*_{\tau}(k))>0$.
We now show that~$H'$ has the necessary properties to prove the lemma.
The following statements can be proven in a tedious yet straightforward manner 
(Figure~\ref{fig:HPrimeLargeMatAsymmetric} provides a visual help):
\begin{enumerate}[(i)]
\item  On $(0,u^*_{+})$ there exists a unique $\bar{u}\in(0,1)$ such that $H'(\bar{u})=0$;
\item  $H':(\bar{u},u^*_{+})\to\RR$ is strictly increasing and tends to infinity at $u^*_+$.
\end{enumerate}
Therefore (i) and (ii) imply that a unique solution to~\eqref{eq:u^*tau} exists satisfying the conditions of the lemma with $u^*_{\tau}(k)\in(\bar{u},u^*_{+})$.
The function $H'$ is strictly positive on $(\bar{u},u^*_{+})$, and hence for large enough $\tau$, 
$u^*_{\tau}(k)$ is strictly increasing and bounded above by $u^*_{+}$, 
and therefore converges to a limit $L\in [\bar{u},u^*_{+}]$. 
If $L\in  [\bar{u},u^*_{+})$, then the continuity of $V'$ and $H'$ and the strict convexity of $V$ implies that
$\lim_{\tau\uparrow \infty}V'(u^*_{\tau}(k))+H'(u^*_{\tau}(k))/\tau=V'(L)<V'(u^*_{+})\leq k$, 
which is a contradiction. 
Therefore $L=u^*_{+}$, which proves Case~(i). 
Cases~(ii) and~(iii) are analogous, and the lemma follows.
\end{proof}

\begin{figure}[h!tb] 
\centering
\mbox{\subfigure[]{\includegraphics[scale=0.7]{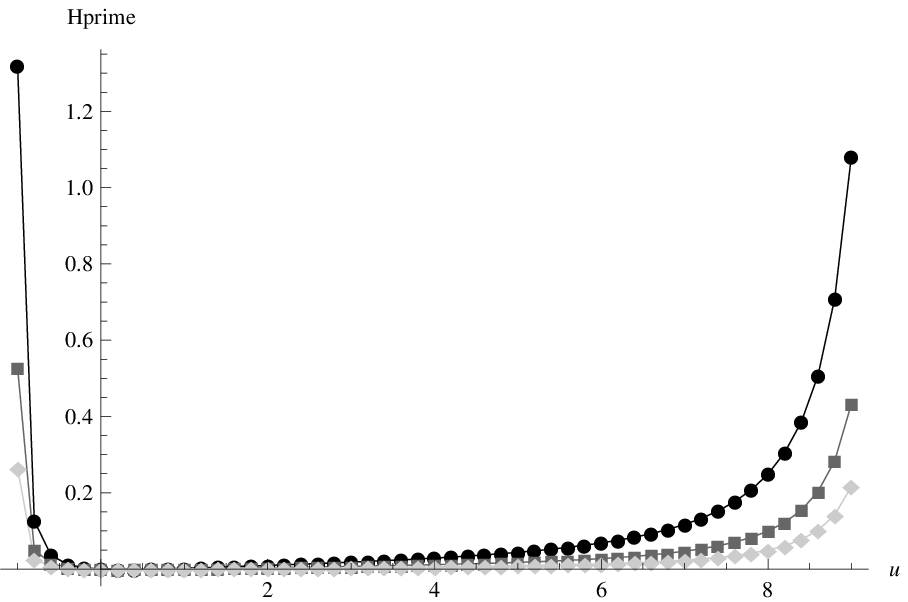}}\quad
\subfigure[]{\includegraphics[scale=0.7]{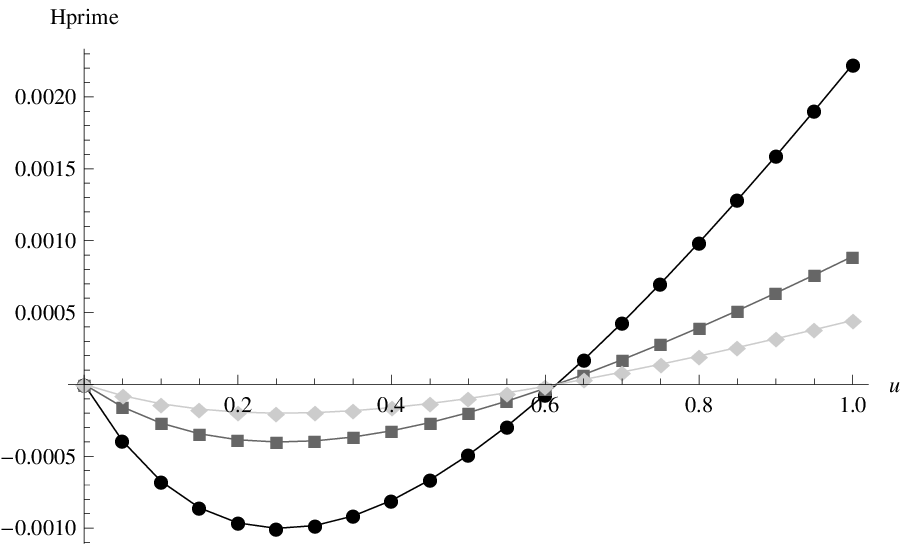}}}
\caption{Plot of $u\mapsto H'(u)/\tau$ for different values of $\tau$.
Circles, Squares and diamonds represent $\tau=2,5,10$.
In (a) $u\in(-1.05,9.72)$ and in (b) $u\in(0,1)$.
The Heston parameters are $v=0.07$, $\theta=0.07$, $\rho=-0.8$, $\xi=0.65$ and $\kappa=1.5$. 
Also $t=1$, $\rho_{-}=-0.56$, $u^*_{+}=9.72$ and $u_{-}=-1.05$. }
\label{fig:HPrimeLargeMatAsymmetric}
\end{figure}

In the following lemma we derive an asymptotic expansion for~$u^*_{\tau}(k)$.
This key result will allow us to derive asymptotics for the characteristic function~$\Phi_{\tau,k,\alpha}$ as well as other auxiliary quantities needed in the analysis.

\begin{lemma}\label{lemma:a-dynamics}
The following expansions hold for $u^*_\tau(k)$ as $\tau$ tends to infinity:
\begin{enumerate}[(i)]
\item
In Regimes $\Rf_2$, $\Rf_{3a}$ and $\Rf_{3b}$,
\begin{enumerate}[(a)]
\item under $\Hh_{\pm}$: 
$
u^*_\tau(k) = u_\pm^*+a_1^\pm(k)\tau^{-1/2} + a_2^\pm(k)\tau^{-1}
 + \mathcal{O}\left(\tau^{-3/2}\right);
$
\item under $\widetilde{\Hh}_{\pm}$:
$
u^*_\tau(k) = u_\pm^* + \widetilde{a}_1^\pm \tau^{-1/3} + \widetilde{a}_2^\pm \tau^{-2/3}
 +\mathcal{O}\left(\tau^{-1}\right);
$
\end{enumerate}
\item
In Regimes $\Rf_{3b}$ and $\Rf_4$,
\begin{enumerate}[(a)]
\item
 For $k>V'(1)$: $u^*_\tau(k)=1-\frac{\mu}{(k-V'(1))}\tau^{-1}+\mathcal{O}(\tau^{-2})$;
\item 
For $k=V'(1)$: 
$u^*_\tau(k) = 1-\tau^{-1/2}\sqrt{\frac{\mu}{V''(1)}}+\mathcal{O}\left(\tau^{-1}\right)$,
\end{enumerate}
\end{enumerate}
with $a_1^{\pm}$, $a_2^{\pm}$ and $a_3^{\pm}$ defined in~\eqref{eq:a} and $u_{\pm}^*$ in~\eqref{eq:DefUpmU*pm}.
\end{lemma}
\begin{proof}
Consider Regime $\Rf_2$ when $\Hh_+$ is in force, i.e. $k>V'(u_+^*)$, and fix such a~$k$.
Existence and uniqueness was proved in Lemma~\ref{lemma:u^*tau} and so we assume the result as an ansatz.
This implies the following asymptotics as~$\tau$ tends to infinity:
\begin{equation}\label{eq:VGammaAsymp}
\left\{
\begin{array}{rl}
V(u^*_\tau(k)) & = \displaystyle V(u^*_+) + \frac{a_1V'(u^*_{+})}{\sqrt{\tau}}
 + \left(\frac{a_1^2 V''(u^*_{+})}{2} + a_2V'(u^*_{+})\right)\frac{1}{\tau}
+\mathcal{O}\left(\frac{1}{\tau^{3/2}}\right), \\
V'(u^*_\tau(k)) & = \displaystyle V'(u^*_{+}) + \frac{a_1V''(u^*_{+})}{\sqrt{\tau}}
 + \left(\frac{a_1^2 V'''(u^*_{+})}{2} + a_2V''(u^*_{+})\right)\frac{1}{\tau}
+\mathcal{O}\left(\frac{1}{\tau^{3/2}}\right), \\
\gamma(u^*_\tau(k)) &
 = \displaystyle \gamma(u^*_{+}) + \frac{a_1\gamma'(u^*_{+})}{\sqrt{\tau}}
  + \left(\frac{a_1^2\gamma''(u^*_{+})}{2} + a_2\gamma'(u^*_{+})\right)\frac{1}{\tau}
+\mathcal{O}\left(\frac{1}{\tau^{3/2}}\right),\\
\gamma'(u^*_{\tau}(k))
 & = \displaystyle \gamma'(u^*_{+}) + \frac{a_1\gamma''(u^*_{+})}{\sqrt{\tau}}
  + \left(\frac{a_1^2\gamma'''(u^*_{+})}{2}+a_2 \gamma''(u^*_{+})\right)\frac{1}{\tau}
+\mathcal{O}\left(\frac{1}{\tau^{3/2}}\right).
\end{array}
\right.
\end{equation}
We substitute this into~\eqref{eq:u^*tau} and solve at each order. 
At the $\tau^{-1/2}$ order we obtain
$
a_1^{+}(k)=\pm\frac{\E^{-\kappa t/2}}{2 \beta _t} \sqrt{\frac{\kappa\theta v}{V'(u_+^*) \left(k-V'(u_+^*)\right)}},
$
which is well-defined since $k-V'(u_+^*)>0$ and $V'(u_+^*)>0$. 
We choose the negative root since we require $u^*_\tau\in(0,u^*_{+})\subset\mathcal{D}_{\infty}^o$ for $\tau$ large enough. 
In a tedious yet straightforward manner we continue the procedure and iteratively solve at each order (the next equation is linear in $a_2$) to derive the asymptotic expansion in the lemma.
The other cases follow from analogous arguments.
\end{proof}

We now derive asymptotic expansions for $\Phi_{\tau,k,\alpha}$.
The expansions will be used in the next section to derive asymptotics for the function $F$ in~\eqref{eq:DF}. 

\begin{lemma}\label{lemma:CharactNonSteepExp}
The following expansions hold as $\tau$ tends to infinity:
\begin{enumerate}[(i)]
\item
In Regimes $\Rf_2$, $\Rf_{3a}$ and $\Rf_{3b}$,
\begin{enumerate}[(a)]
\item under $\Hh_{\pm}$: 
$
\Phi_{\tau,k,3/4}(u) = \E^{-\zeta^2_\pm(k)u^2/2}\left(1+\max(1,u^s)\mathcal{O}\left(\tau^{-1/4}\right)\right);
$
\item under $\widetilde{\Hh}_{\pm}$:
$
\Phi_{\tau,k,1/2}(u) = 
\E^{-3 V''\left(u_\pm^*\right)u^2/2}\left(1+\max(1,u^s)\mathcal{O}\left(\tau^{-1/6}\right)\right);
$
\end{enumerate}
\item
In Regimes $\Rf_{3b}$ and $\Rf_4$,
\begin{enumerate}[(a)]
\item
 For $k>V'(1)$: $
\Phi_{\tau,k,1}\left(u\right) = \exp\left(-\I u (k-V'(1) ) - \frac{u^2 V''(1)}{2\tau}\right)\left(1-\I u \frac{(k-V'(1))}{\mu}\right)^{-\mu}
(1+\max(1,u^s)\mathcal{O}(\tau^{-1}));
$
\item 
For $k=V'(1)$: $
\Phi_{\tau,k,1/2}\left(u\right) = \exp\left(-\I u \sqrt{\mu V''(1)} - \frac{u^2 V''(1)}{2}\right)\left(1-\I u \sqrt{\frac{V''(1) }{\mu}}\right)^{-\mu}
(1+\max(1,u^s)\mathcal{O}(\tau^{-1/2})),
$
\end{enumerate}
\end{enumerate}
for some integer $s$ different from one line to the other.
Recall that $\Phi_{\tau,k,\alpha}$ is defined in~\eqref{eq:CharacNonSteep} and $\zeta^2_{\pm}$ in~\eqref{eq:variance}.
Furthermore, as $\tau$ tends to infinity the remainders in (i) and (ii)(b) are uniform in $u$ for $|u|<\tau^{1/6}$ and the remainder in (ii)(a) is uniform in $u$ for $|u|<\tau^{2/3}$.
\end{lemma}

\begin{remark}\label{rem:limitprops}\ 
\begin{enumerate}[(i)]
\item In Case (i)(a), $Z_{\tau,k,3/4}$ converges weakly to a centred Gaussian 
with variance $\zeta^2_{\pm}(k)$ when $\Hh_\pm$ holds.
\item
In Case (i)(b), $Z_{\tau,k,1/2}$ converges weakly a centred Gaussian 
with variance $3V''(u_{+})$ when $\widetilde{\Hh}_\pm$ holds.
\item
In Case(ii)(a), $Z_{\tau,k,1}$ converges weakly to the zero-mean random variable $\Xi-\gamma$, 
where $\gamma:=k-V'(1)$ and $\Xi$ is a Gamma random variable with shape parameter $\mu$  
and scale parameter $\beta:=(k-V'(1))/\mu$.
Lemma~\ref{lem:Gammaasymplargecorrel} implies that the limiting characteristic function satisfies $\int_{-\infty}^{\infty}\left(1-\I u\beta\right)^{-\mu}\E^{-V''(1)u^2/(2\tau)}u^j\D u=\mathcal{O}(1)$ for any $j\in\mathbb{N}\cup\{0\}$.
\item 
In Case(ii)(b), $Z_{\tau,k,1/2}$ 
converges weakly to the zero-mean random variable $\Psi+\Xi$,
where $\Psi$ is Gaussian with mean $-\sqrt{\mu V''(1)}$ and variance $V''(1)$ 
and $\Xi$ is Gamma-distributed with shape~$\mu$ and scale~$\sqrt{V''(1)/\mu}$.
\end{enumerate}
\end{remark}

We now prove Case~(i)(a) in Regime $\Rf_2$, as the proofs in all other cases are similar. 
In the forthcoming analysis we will be interested in the asymptotics of the function~$e_{\tau}$ defined by
\begin{equation}\label{eq:etau}
e_{\tau}(k) \equiv \sqrt{\tau}\left(\kappa\theta -2 \beta _t V(u^*_\tau(k))\right).
\end{equation}
Under $\Rf_2$, in Case (i)(a), $\left(\kappa\theta -2 \beta _t V(u^*_\tau)\right)$ tends to zero as $\tau$ tends to infinity,
so that it is not immediately clear what happens to $e_{\tau}$ for large $\tau$. 
But the asymptotic behaviour of $V(u^*_\tau)$ in~\eqref{eq:VGammaAsymp} and the definition~\eqref{eq:etau} yield the following result:
\begin{lemma}\label{lemma:e-dynamics}
Assume $\Rf_2$ and $\Hh_+$.
Then the expansion 
$e_{\tau}(k) = e_0^+(k) + e_1^+(k)\tau^{-1/2}+\mathcal{O}\left(\tau^{-1}\right)$
holds as~$\tau$ tends to infinity, 
with  $e_0$ and $e_1$ defined in~\eqref{eq:e} and $u_{\pm}^*$ in~\eqref{eq:DefUpmU*pm}.
\end{lemma}

\begin{proof}[Proof of Lemma~\ref{lemma:CharactNonSteepExp}]
Consider Regime $\Rf_2$ when $\Hh_+$ is in force, i.e. $k>V'(u_+^*)$, and fix such a~$k$, 
and for ease of notation drop the superscripts and $k$-dependence.
Lemma~\ref{lem:charactsimp} yields
\begin{equation}
\log\Phi_{\tau,k}(u) = -\I u k \tau^{1/4} +\tau\left(V\left(\frac{\I u}{ \tau^{3/4}} + u^*_{\tau}\right)-V\left( u^*_{\tau}\right)\right)+H\left(\frac{\I u }{\tau^{3/4}} + u^*_{\tau}\right)-H\left( u^*_{\tau}\right)
+\mathcal{O}(\tau^{-1/4}).
\label{eq:phiztaunonsteepderiv}
\end{equation}
Using Lemma~\ref{lemma:a-dynamics}, we have the Taylor expansion (similar to~\eqref{eq:VGammaAsymp})
\begin{align}\label{eq:Vustarttaui}
 V\left(u^*_{\tau}+\I u/\tau^{3/4}\right)=\frac{\kappa\theta}{2\beta_t}+\frac{a_1V'}{\sqrt{\tau}}+\frac{\I u V'}{\tau^{3/4}}+\left(\frac{V''a_1^2}{2}+V' a_2\right)\frac{1}{\tau}+\frac{\I u a_1 V''}{\tau^{5/4}}+\mathcal{O}\left(\frac{1}{\tau^{3/2}}\right),
\end{align}
as $\tau$ tends to infinity, where $V$, $V'$ and $V''$ are evaluated at $u^*_{+}$. 
Using~\eqref{eq:VGammaAsymp} we further have
\begin{align}
 V\left(u^*_{\tau}+\I u/\tau^{3/4}\right)-V\left(u^*_{\tau}\right)
 & = \frac{\I u V'(u^*_{+})}{\tau^{3/4}}+\frac{\I u a_1 V''(u^*_{+})}{\tau^{5/4}}+\mathcal{O}\left(\frac{1}{\tau^{3/2}}\right),
\label{eq:Vustarttauidiff}\\
\gamma\left(u^*_{\tau}+\I u/\tau^{3/4}\right)
 & = \gamma(u^*_{+})+\frac{a_1\gamma'(u^*_{+})}{\sqrt{\tau}}+\frac{\I u \gamma'(u^*_{+})}{\tau^{3/4}}+\mathcal{O}\left(\frac{1}{\tau}\right).
\label{eq:Gammaustartauuidiff}
\end{align}
We now study the behaviour of 
$H\left(\I u /\tau^{3/4} + u^*_{\tau}\right)$, where $H$ is defined in~\eqref{eq:VandH}.
Using Lemma~\ref{lemma:e-dynamics} and the expansion~\eqref{eq:Vustarttauidiff} for large $\tau$,
we first note that
\begin{align}\label{eq:den}
e_{\tau} - 2 \beta _t \sqrt{\tau} \left[V(u^*_{\tau}+\frac{\I u}{\tau^{3/4}})-V(u^*_{\tau})\right]
 & = e_0-\frac{2\beta_t\I u V'}{\tau^{1/4}}+\frac{e_1}{\sqrt{\tau}}
-\frac{2\beta_t \I u a_1 V''}{\tau^{3/4}}+\mathcal{O}\left(\frac{1}{\tau}\right),
\end{align}
with $e_{\tau}$ defined in~\eqref{eq:etau}.
Together with~\eqref{eq:Vustarttaui}, this implies
\begin{align}\label{eq:HFirstTerm1}
\frac{v \E^{-\kappa t}V(u^*_{\tau}+\I u/ \tau^{3/4})}{\kappa\theta -2 \beta _t V(u^*_{\tau}+\I u/ \tau^{3/4})}
& = \frac{\sqrt{\tau}v \E^{-\kappa t} V(u^*_{\tau}+\I u/ \tau^{3/4})}{e_{\tau}-2 \beta _t \sqrt{\tau} \left(V(u^*_{\tau}+\I u/ \tau^{3/4})-V(u^*_{\tau})\right)}\nonumber\\
& = \frac{\kappa\theta  v \E^{-\kappa t}\sqrt{\tau}}{2 e_0 \beta _t}+\frac{\I \kappa\theta  u v e^{-\kappa t} V' \tau^{1/4}}{e_0^2 }+v \E^{-\kappa t} \left(\frac{a_1 V'}{e_0}-\frac{e_1 \kappa\theta }{2 e_0^2 \beta _t}\right)  -\frac{\zeta_{+}^2u^2}{2}+\mathcal{O}\left(\frac{1}{\tau^{1/4}}\right),
\end{align}
with $\zeta_{+}$ defined in~\eqref{eq:variance}. 
Substituting $e_0$ in~\eqref{eq:e} into the second term in~\eqref{eq:HFirstTerm1} we find
\begin{equation}\label{eq:simplify}
\frac{\I \kappa\theta  u v e^{-\kappa t} V' }{e_0^2 }=\I u \left(k-V'\right).
\end{equation}
Following a similar procedure using $e_\tau$ we establish  for large $\tau$ that
\begin{align}\label{eq:HFirstTerm2}
&\frac{v \E^{-\kappa t}V(u^*_{\tau})}{\kappa\theta -2 \beta _t V(u^*_{\tau})}=\frac{\kappa\theta  v \E^{-\kappa t}\sqrt{\tau}}{2 e_0 \beta _t}+v \E^{-\kappa t} \left(\frac{a_1 V'}{e_0}-\frac{e_1 \kappa\theta }{2 e_0^2 \beta _t}\right) +\mathcal{O}\left(\frac{1}{\sqrt{\tau}}\right),
\end{align}
and combining~\eqref{eq:HFirstTerm1},~\eqref{eq:simplify} and~\eqref{eq:HFirstTerm2} we find that
\begin{equation}\label{eq:HFirstTerm12}
\frac{V(u^*_{\tau}+\I u/ \tau^{3/4})v \E^{-\kappa t}}{\kappa\theta -2 \beta _t V(u^*_{\tau}+\I u/ \tau^{3/4})}
 - \frac{V(u^*_{\tau})v \E^{-\kappa t}}{\kappa\theta -2 \beta _t V(u^*_{\tau})}  
=\I u \left(k-V'\right)\tau^{1/4}-\frac{\zeta_+^2u^2}{2} +\mathcal{O}\left(\frac{1}{\tau^{1/4}}\right).
\end{equation}
We now analyse the second term of $\exp(H(\I u /\tau^{3/4} + u^*_{\tau})-H(u^*_{\tau}))$. We first re-write this term as 
\begin{align}
&\exp\left(-\mu\log \left(\frac{\kappa\theta-2 \beta _t V(u^*_{\tau}+\I u/ \tau^{3/4}) }{\kappa\theta \left(1-\gamma \left(u^*_{\tau}+\I u/ \tau^{3/4}\right)\right)}\right)
+\mu\log \left(\frac{\kappa\theta-2 \beta _t V(u^*_{\tau}) }{\kappa\theta \left(1-\gamma \left(u^*_{\tau}\right)\right)}\right)\right) \\ \nonumber
&=\left(\left(\frac{\kappa\theta-2 \beta _t V(u^*_{\tau}+\I u/ \tau^{3/4})}{\kappa\theta-2 \beta _t V(u^*_{\tau})}\right)\left(\frac{1-\gamma \left(u^*_{\tau}+\I u/ \tau^{3/4}\right)}{1-\gamma \left(u^*_{\tau}\right)}\right)^{-1}\right)^{-\mu},
\end{align}
and deal with each of the multiplicative terms separately. For the first term we re-write it as
\begin{align}
\frac{\kappa\theta-2 \beta _t V(u^*_{\tau}+\I u/ \tau^{3/4})}{\kappa\theta-2 \beta _t V(u^*_{\tau})}=\frac{e_{\tau}-2 \beta _t \sqrt{\tau} \left(V(u^*_{\tau}+\I u/ \tau^{3/4})-V(u^*_{\tau})\right)}{e_\tau},
\end{align}
and then we use the asymptotics of $e_{\tau}$ in~\ref{lemma:e-dynamics} and equation~\eqref{eq:den}
to find that as~$\tau$ tends to infinity,
\begin{align}
\frac{\kappa\theta-2 \beta _t V(u^*_{\tau}+\I u/ \tau^{3/4})}{\kappa\theta-2 \beta _t V(u^*_{\tau})}=1+ \mathcal{O}\left(\frac{1}{\tau^{1/4}}\right).
\end{align}
For the second term we use the asymptotics in~\eqref{eq:VGammaAsymp} and~\eqref{eq:Gammaustartauuidiff} to find that for large $\tau$
$$
\left(\frac{1-\gamma \left(u^*_{\tau}+\I u/ \tau^{3/4}\right)}{1-\gamma \left(u^*_{\tau}\right)}\right)^{-1}=\left(\frac{1-\left(\gamma+a_1\gamma'/\sqrt{\tau}+\I u \gamma'/\tau^{3/4}+\mathcal{O}(1/\tau)\right)}{1-\left(\gamma+a_1\gamma'/\sqrt{\tau}+\mathcal{O}(1/\tau)\right)}\right)^{-1} 
=1+\mathcal{O}(1/\tau^{3/4}).
$$
It then follows that for the second term of $\exp(H(\I u /\tau^{3/4} + u^*_{\tau})-H(u^*_{\tau}))$ that for large $\tau$ we have
\begin{equation}\label{eq:HSecTerm12}
\exp\left(-\mu\log \left(\frac{\kappa\theta-2 \beta _t V(u^*_{\tau}+\I u/ \psi_\tau) }{\kappa\theta \left(1-\gamma \left(u^*_{\tau}+\I u/ \psi_\tau\right)\right)}\right)
+\mu\log \left(\frac{\kappa\theta-2 \beta _t V(u^*_{\tau}) }{\kappa\theta \left(1-\gamma \left(u^*_{\tau}\right)\right)}\right)\right) = 1+\mathcal{O}\left(\frac{1}{\tau^{1/4}}\right).
\end{equation}
Further as $\tau$ tends to infinity, the equality~\eqref{eq:Vustarttauidiff} implies
\begin{align}\label{eq:Vustarttauidifftau}
\tau\left(V(u^*_{\tau}+\I u/ \tau^{3/4})-V(u^*_{\tau})\right)
=\I u V'(u^*)\tau^{1/4}+\mathcal{O}(\tau^{-1/4}).
\end{align}
Combining~\eqref{eq:HFirstTerm12},~\eqref{eq:HSecTerm12} and~\eqref{eq:Vustarttauidifftau} into~\eqref{eq:phiztaunonsteepderiv} completes the proof.
The proof of the uniformity of the remainders and the existence of the integer~$s$
follow the same lines as the proof of~\cite[Lemmas 7.1, 7.2]{BR02}.
\end{proof}

In order to derive complete asymptotic expansions we still need to derive expansions for $D$ and $F$ in~\eqref{eq:DF}.
This is the purpose of this section.
We first derive an expansion for $D$ which gives the leading-order decay 
of large-maturity out-of-the-money options:

\begin{lemma} \label{lemma:V*Asymptotics}
The following expansions hold as $\tau$ tends to infinity:
\begin{enumerate}[(i)]
\item
In Regimes $\Rf_2$, $\Rf_{3a}$ and $\Rf_{3b}$, 
\begin{enumerate}[(a)]
\item
under $\Hh_\pm$: 
$
D(\tau,k) = 
\exp\left(-\tau (V^*(k)-k)+\sqrt{\tau}c_0^{\pm}(k)+c_1^{\pm}(k)\right)
\tau^{\mu/2}c_2^{\pm}(k)(1+\mathcal{O}(\tau^{-1/2}));
$
\item
under $\widetilde{\Hh}_\pm$:
$
D(\tau,k) = 
\exp\left(-\tau (V^*(k)-k)+\tau^{1/3}c_0^\pm + c_1^{\pm}\right)\tau^{\mu/3}c_2^{\pm}
(1+\mathcal{O}(\tau^{-1/3}));
$
\end{enumerate}
\item
In Regimes $\Rf_{3b}$ and $\Rf_4$,
\begin{enumerate}[(a)]
\item
 For $k>V'(1)$: $
D(\tau,k)=
\exp\left(-\tau (V^*(k)-k)+\mu+g_0\right)
\left(\frac{2(k-V'(1))(\kappa-\rho\xi )^2}{ \left(\kappa\theta -2 V(1) \beta _t\right)}\right)^{\mu}
\tau^{\mu}
(1+\mathcal{O}(\tau^{-1}));
$
\item 
For $k=V'(1)$: $
D(\tau,k)=
\exp\left(-\tau (V^*(k)-k)+\mu/2+g_0\right)
\left(\frac{2(\kappa-\rho\xi )^2 \sqrt{V''(1) \mu}}{\left(\kappa\theta -2 V(1) \beta _t\right)}\right)^{\mu}
\tau^{\mu/2}(1+\mathcal{O}(\tau^{-1/2})).
$
\end{enumerate}
\end{enumerate}
where $c_0$, $c_1$ and $c_2$  in~\eqref{eq:c0}, $g_0$ in~\eqref{eq:c0bound} and $V^*$ is characterised explicitly in Lemma~\ref{lemma:V*Characterisation}.
\end{lemma}
\begin{proof}
Consider Regime $\Rf_2$ in Case(i)(a) (namely when $\Hh_+$ holds), and again for ease of notation drop the superscripts and $k$-dependence.
We now use Lemma~\ref{lemma:a-dynamics} and~\eqref{eq:VGammaAsymp} to write for large $\tau$:
\begin{align}\label{eq:step1asymp}
\E^{-\tau\left(ku^*_{\tau}-V(u^*_{\tau})\right)}
 & = \exp\left[-\tau(k u^*_{+}-V(u^*_{+}))-\sqrt{\tau}a_1(k-V')+r_0-a_2 k +\mathcal{O}(\tau^{-1/2})\right] \\ \nonumber
 & = \E^{-\tau V^*(k)-\sqrt{\tau}a_1(k-V')+r_0-a_2 k}\left[1+\mathcal{O}(\tau^{-1/2})\right],
\end{align}
with $r_0 := \frac{1}{2}V''a_1^2 + V' a_2$ and where we have used the characterisation of $V^*$ 
given in Lemma~\ref{lemma:V*Characterisation}. 
We now study the asymptotics of $H(u^*_\tau)$. 
Using the definition of $e_\tau$ in~\eqref{eq:etau} we write
\begin{equation}\label{eq:step2asymp}
 \E^{{H(u^*_{\tau})}}
 = \exp\left(\frac{V(u^*_{\tau})v \E^{-\kappa t}}{\kappa\theta -2 \beta _t V(u^*_{\tau})}
\right)
\left[\frac{\kappa\theta-2 \beta _t V(u^*_{\tau})}{\kappa\theta \left(1-\gamma\left(u^*_{\tau}\right)\right)}\right]^{-\mu} 
 = \tau^{\frac{\mu}{2}}
 \exp\left(\frac{V(u^*_{\tau})v \E^{-\kappa t}}{\kappa\theta -2 \beta _t V(u^*_{\tau})}
\right)
\left[\frac{e_{\tau}}{\kappa\theta \left(1-\gamma \left(u^*_{\tau}\right)\right)}\right]^{-\mu},
\end{equation}
and deal with each of these terms in turn. 
Now by~\eqref{eq:HFirstTerm2} we have, as $\tau$ tends to infinity,
\begin{align}\label{eq:step3asymp}
&\frac{v \E^{-\kappa t}V(u^*_{\tau})}{\kappa\theta -2 \beta _t V(u^*_{\tau})}=\frac{\kappa\theta  v \E^{-\kappa t}\sqrt{\tau}}{2 e_0 \beta _t}+v \E^{-\kappa t} \left(\frac{a_1 V'}{e_0}-\frac{e_1 \kappa\theta}{2 e_0^2 \beta _t}\right) +\mathcal{O}\left(\frac{1}{\sqrt{\tau}}\right).
\end{align}
Using the asymptotics of $e_{\tau}$ given in Lemma~\ref{lemma:e-dynamics} and those of~$\gamma$ in~\eqref{eq:VGammaAsymp} we find 
\begin{align}\label{eq:step4asymp}
 \left(\frac{e_{\tau}}{\kappa\theta\left(1-\gamma \left(u^*_{\tau}\right)\right)}\right)^{-\mu}
  = \left(\frac{e_{0}+e_1/\sqrt{\tau}+\mathcal{O}\left(1/\tau\right)}{\kappa\theta \left(1-\gamma\right)+\kappa\theta a_1\gamma'/\sqrt{\tau}+\mathcal{O}\left(1/\tau\right)}\right)^{-\mu} 
=
\left(\frac{\kappa\theta\left(1-\gamma \right)}{e_0}\right)^{\mu} 
\left(1+\mathcal{O}\left(\frac{1}{\sqrt{\tau}}\right)\right).
\end{align}
Using the definition of $e_0$ in~\eqref{eq:e}, note the simplification
$-a_1(k-V')+\frac{\kappa\theta v \E^{-\kappa t}}{2 e_0 \beta _t}=-2 a_1 (k-V')$.
Combining this,~\eqref{eq:step1asymp},~\eqref{eq:step2asymp},~\eqref{eq:step3asymp} and ~\eqref{eq:step4asymp} we find that
$$
D(\tau,k):=
\E^{-\tau\left(k(u^*_{\tau}-1)-\Lambda_{\tau}^{(t)}\left(u^*_{\tau}\right)\right)}
  = 
\exp\left(-\tau (V^*(k)-k)+\sqrt{\tau}c_0^{+}+c_1^{+}\right)\tau^{\mu/2}c_2^{+}
(1+\mathcal{O}(\tau^{-1/2})),
$$
with $c_0^+$, $c_1^+$ and $c_2^+$  in~\eqref{eq:c0}. 
All other cases follows in an analogous fashion and this completes the proof.
\end{proof}

In Lemma~\ref{lemma:FourierAsymptotics} below we provide asymptotic expansions 
for the function $F$ in~\eqref{eq:DF}.
However, we first need the following technical result, the proof of which can be found in~\cite[Lemma 7.3]{BR02}.
Let $p$ denote the density of a Gamma random variable with shape $\lambda$ and scale $\nu$, 
and $\widehat{p}$ the corresponding characteristic function:
\begin{equation}\label{eq:gammadensitycharact}
p(x) \equiv \frac{1}{\Gamma(\lambda)\nu^{\lambda}}x^{\lambda-1}\E^{-x/\nu}\ind_{\{x>0\}}, 
\qquad \widehat{p}(u) \equiv (1-\I \nu u)^{-\lambda}.
\end{equation}
\begin{lemma}\label{lem:Gammaasymplargecorrel}
The following expansion holds as $\tau$ tends to infinity:
$$
\int_{\RR} \exp\left({-\I \gamma u - \frac{\sigma^2u^2}{2\tau} }\right)u^{\beta} \widehat{p}(\gamma u)\D u
=\sum_{r=0}^{q}\frac{2\pi \sigma^{2r}}{\I^{\beta} \gamma^{2r+\beta+1}2^{r}r! \tau^{r}}p^{(2r+\beta)}(1)+\mathcal{O}\left(\frac{1}{\tau^{q+1}}\right),
$$
with $\gamma,\nu,\lambda\in\RR^*_{+}$, 
$\beta\in\mathbb{N}\cup\{0\}$, $q\in\mathbb{N}$ and $p^{(n)}$ denoting the 
$n$-th derivative of the Gamma density~$p$.
\end{lemma}

\begin{lemma}\label{lemma:FourierAsymptotics}
The following expansions hold as $\tau$ tends to infinity
(with $\zeta_{\pm}$ in~\eqref{eq:variance} and $u^*_{\pm}$ in~\eqref{eq:DefUpmU*pm}):
\begin{enumerate}[(i)]
\item
In Regimes $\Rf_2$, $\Rf_{3a}$ and $\Rf_{3b}$, 
\begin{enumerate}[(a)]
\item under $\Hh_\pm$: 
$
F(\tau,k,3/4) = 
\frac{\tau^{-3/4}}{\zeta_\pm(k) u^*_+(u^*_{\pm}-1)\sqrt{2\pi}}(1+\mathcal{O}(\tau^{-1/2}))
$;
\item under $\widetilde{\Hh}_\pm$: 
$
F(\tau,k,1/2) = 
\frac{\tau^{-1/2}}{u^*_{\pm}(u^*_{\pm}-1)\sqrt{6\pi V''(u^*_{\pm})}}(1+\mathcal{O}(\tau^{-1/3}))
$;
\end{enumerate}
\item
In Regimes $\Rf_{3b}$ and $\Rf_4$,
\begin{enumerate}[(a)]
\item
 For $k>V'(1)$: 
$F(\tau,k,1)
=-\frac{\E^{-\mu}\mu^{\mu}}{\Gamma(1+\mu)}(1+\mathcal{O}(\tau^{-1}))$;
\item 
For $k=V'(1)$: 
$F(\tau,k,1/2)
=-\frac{\E^{-\mu/2}(\mu/2)^{\mu/2}}{2\Gamma(1+\mu/2)}(1+\mathcal{O}(\tau^{-1/2}))$.
\end{enumerate}
\end{enumerate}
\end{lemma}

\begin{proof}
Again, we only consider here Regime $\Rf_2$ under $\Hh_+$ in Case (i)(a).
Using the asymptotics of $u^*_{\tau}$ given in Lemma~\ref{lemma:a-dynamics}, 
we can Taylor expand for large~$\tau$ to obtain
$
 \overline{C(\tau,k,3/4)}
=\frac{\tau^{-3/4}}{(u_+^*-1) u_+^*}(1+\mathcal{O}(\tau^{-1/2})),
$
where the remainder $\mathcal{O}(\tau^{-1/2})$ is uniform in $u$ as soon as $u=\mathcal{O}(\tau^{3/4})$.
Combining this with the characteristic function asymptotics in Lemma~\ref{lemma:CharactNonSteepExp} we find that for large~$\tau$,
$
F(\tau,k,3/4)
=\frac{1}{\tau ^{3/4} \left(u_+^*-1\right) u_+^*}\int_{\RR}\exp\left(-\frac{\zeta_+^2(k)u^2}{2}\right)(1+\mathcal{O}(\tau^{-1/4}))\D u.
$
Using Lemma~\ref{lem:expsmalllargetime}, there exists $\beta>0$ such that as $\tau$ tends to infinity we can write this integral as
\begin{align*}
\int_{-\infty}^{\infty}\exp\left(-\frac{\zeta_+^2(k)u^2}{2}\right)\left(1+\mathcal{O}(\tau^{-1/4})\right)\D u
&=\int_{-\tau^{3/4}}^{\tau^{3/4}}\exp\left(-\frac{\zeta_+^2(k)u^2}{2}\right)
\left(1+\mathcal{O}(\tau^{-1/4})\right)\D u 
+\mathcal{O}(\E^{-\beta \tau}) \\
&=\int_{-\tau^{3/4}}^{\tau^{3/4}}\exp\left(-\frac{\zeta_+^2(k)u^2}{2}\right)\D u
\left(1+\mathcal{O}(\tau^{-1/4})\right)
+\mathcal{O}(\E^{-\beta \tau}) \\
&=\int_{\RR}\exp\left(-\frac{\zeta_+^2(k)u^2}{2}\right)\D u \left(1+\mathcal{O}(\tau^{-1/4})\right)
 = \frac{\sqrt{2\pi}}{|\zeta(k)|}\left(1+\mathcal{O}(\tau^{-1/4})\right).
\end{align*}
The second line follows from Lemma~\ref{lemma:CharactNonSteepExp} and in the third line we have used that the tail estimate for the Gaussian integral is exponentially small and absorbed this into the remainder $\mathcal{O}(\tau^{-1/4})$.
By extending the analysis to higher order the $\mathcal{O}(\tau^{-1/4})$ term is actually zero and the next non-trivial term is $\mathcal{O}(\tau^{-1/2})$.
For brevity we omit the analysis and we give the remainder as $\mathcal{O}(\tau^{-1/2})$ in the lemma.
Case (i)(b) follows from analogous arguments to above and we now move onto Case (ii)(a).
Using the asymptotics of~$u^*_{\tau}$ in Lemma~\ref{lemma:a-dynamics} we have
$
\overline{C(\tau,k,1)}
 = -\left(\frac{\mu}{\nu(k)} - \I u\right)^{-1}+\mathcal{O}(\tau^{-1})
 = \frac{-\nu(k)}{\mu}\left(1-\frac{\I u \nu(k)}{\mu}\right)^{-1}+\mathcal{O}(\tau^{-1}),
$
where we set 
$\nu(k):=k-V'(1)$ and the remainder $\mathcal{O}(\tau^{-1})$ is uniform in $u$ as soon as $u=\mathcal{O}(\tau)$.
Using the characteristic function asymptotics in Lemma~\ref{lemma:CharactNonSteepExp} and Lemma~\ref{lem:expsmalllargetime}, there exists $\beta>0$ such that as $\tau$ tends to infinity:
\begin{align}
F(\tau,k,1)
 & =\frac{-\nu}{2\pi \mu}\int_{-\tau}^{\tau} \exp\left(-\I u \nu - \frac{u^2 V''(1)}{2 \tau}\right)
 \left(1-\frac{\I u \nu}{\mu}\right)^{-\mu}
 \left[\left(1-\frac{\I u \nu}{\mu}\right)^{-1}+\mathcal{O}\left(\tau^{-1}\right)\right]  \D u
 + \mathcal{O}(\E^{-\beta \tau}) \nonumber \\ 
 & = \frac{-\nu}{2\pi \mu}\int_{-\tau}^{\tau} \exp\left(-\I u \nu - \frac{u^2 V''(1)}{2 \tau}\right)
\left(1-\frac{\I u \nu}{\mu}\right)^{-1-\mu}  \D u \left[1+\mathcal{O}\left(\tau^{-1}\right)\right]+\mathcal{O}(\E^{-\beta \tau}).\label{eq:largecorrelextend}
\end{align}
The second line follows from Lemma~\ref{lemma:CharactNonSteepExp} and Remark~\ref{rem:limitprops}(iii).
Further we note that
\begin{align*}
\left|\int_{|u|>\tau} \exp\left(-\I u \nu - \frac{u^2 V''(1)}{2 \tau}\right)
\left(1-\frac{\I u \nu}{\mu}\right)^{-1-\mu}  \D u \right|
&\leq \tau \int_{|z|>1} \E^{-\frac{1}{2}\tau z^2 V''(1)}
\left(1+\frac{ z^2 \tau^2 \nu^2}{\mu^2}\right)^{-1-\mu}  \D z \\
&\leq \tau \int_{|z|>1} \E^{-\frac{1}{2}\tau z^2 V''(1)}  \D z
=\mathcal{O}\left(\E^{-\Delta\tau}\right), 
\end{align*}
for some $\Delta>0$ as $\tau$ tends to infinity.
Combining this with~\eqref{eq:largecorrelextend} we can write
\begin{align*}
F(\tau,k,1)
& = \frac{-\nu}{2\pi \mu}\int_{-\infty}^{\infty} \exp\left(-\I u \nu - \frac{u^2 V''(1)}{2 \tau}\right)
\left(1-\frac{\I u \nu}{\mu}\right)^{-1-\mu}  \D u \left[1+\mathcal{O}\left(\tau^{-1}\right)\right],\\
 & = \left(-\frac{\E^{-\mu}\mu^{\mu}}{\Gamma(1+\mu)}
+\mathcal{O}\left(\tau^{-1}\right)\right)\left[1+\mathcal{O}\left(\tau^{-1}\right)\right],
\end{align*}
where we have absorbed the exponential remainder into $\mathcal{O}(\tau^{-1})$, 
and where the second line follows from Lemma~\ref{lem:Gammaasymplargecorrel}.
We now prove (ii)(b).
Using the asymptotics of $u^*_{\tau}$ for large $\tau$ in Lemma~\ref{lemma:a-dynamics}, we obtain
$\overline{C(\tau,k,1/2)} = \frac{1}{a_1(1+\I u/a_1)}+\mathcal{O}(\tau^{-1/2})$,
with $a_1=-\sqrt{\frac{\mu}{V''(1)}}$ and where the remainder $\mathcal{O}(\tau^{-1/2})$ is uniform in $u$ as soon as $u=\mathcal{O}(\tau^{1/2})$.
Using the characteristic function asymptotics in  Lemma~\ref{lemma:CharactNonSteepExp} and analogous arguments as above we have the following expansion for large~$\tau$:
$$
F(\tau,k,1/2)
=\frac{1}{2\pi a_1}
\int_{\RR} \frac{\exp\left(\I u a_1 V''(1) -\frac{1}{2}u^2 V''(1)\right)}{(1+\I u /a_1)^{1+\mu}}\D u \left(1+\mathcal{O}\left(\tau^{-1/2}\right)\right).
$$
Let $n$ and $\widehat{n}$ denote the Gaussian density and characteristic function with zero mean and variance $V''(1)$.
Using~\eqref{eq:gammadensitycharact}, we have
$$
\int_{\RR}\E^{-\I \omega u}\widehat{n}(u)\widehat{p}(u)\D u 
=2\pi \mathcal{F}^{-1}(\widehat{n}(u)\widehat{p}(u))(\omega)
=2\pi \mathcal{F}^{-1}(\mathcal{F}(n * p))
=2 \pi \int_{0}^{\infty} n(\omega-y) p(y) \D y,
$$
so that
$$
\frac{1}{2\pi a_1}\int_{\RR} \frac{\exp\left(\I u a_1 V''(1) - \frac{1}{2}u^2 V''(1)\right)}{(1+\I u /a_1)^{1+\mu}}\D u
=\frac{1}{a_1}\int_{0}^{\infty}n(-a_1 V''(1) - y) p(y) \D y.
$$
This integral can now be computed in closed form and the result follows after simplification 
using the definition of $a_1$ and the duplication formula for the Gamma function.
\end{proof}

\subsection{Asymptotics in the case of non-existence of the limiting Fourier transform}\label{sec:V0V1limFT}
In this section, we are interested in the cases where $k\in\{V'(0),V'(1)\}$ whenever $\Hh_0$ is in force, 
which corresponds to all the regimes except $\Rr_{3b}$ and $\Rr_4$ at $V'(1)$.
In these cases, the limiting  Fourier transform is undefined at these points.
We show here however that the methodology of Section~\ref{sec:Genmethlargetime} 
can still be applied, 
and we start by verifying Assumption~\ref{assump:ustartaueq}.
The following quantity will be of primary importance: 
\begin{equation}\label{eq:Upsilon}
\Upsilon(a) := 1 + \frac{a \rho\xi}{\kappa-\rho\xi}\E^{\kappa t},
\end{equation}
for $a\in\{0,1\}$, 
and it is straightforward to check that $\Upsilon$ is well defined whenever $\Hh_0$ is in force.
\begin{lemma} \label{lem:v0v1ustarexist}
Let $a\in\{0,1\}$ and assume that $v \ne \theta\Upsilon(a)$.
Then, whenever $\Hh_0$ holds, Assumption~\ref{assump:ustartaueq} is satisfied 
with $\mathcal{A}=\{V'(a)\}$ and $u_{\infty}^*=a$.
Additionally, if $v<\theta\Upsilon(a)$, then there exists $\tau_1^*>0$ such that $u_{\tau}^*(k)<0$ 
if $a=0$ and $u_{\tau}^*(k)>1$ if $a=1$ for all $\tau>\tau_1^*$, and if $v>\theta\Upsilon(a)$, then there exists $\tau_1^*>0$ such that $u_{\tau}^*(k)\in (0,1)$ for all $\tau>\tau_1^*$;
\end{lemma}

\begin{proof}
Recall that the function $H$ is defined in~\eqref{eq:VandH}.
We first prove the lemma in the case $a=0$, in which case $\Upsilon(0) = 1$.
Note that $H'(0)>0(<0)$ if and only if $v/\theta<1(>1)$ and $H'(0)=0$ if and only if $v=\theta$.
Now let $k=V'(0)$ and $v<\theta$ and consider the equation $H'(u)/\tau=V'(0)-V'(u)$.
Since $H'$ is continuous $H'$ is strictly positive in some neighbourhood of zero. 
In order for the right-hand side to be positive we require our solution to be in $(-\delta_0,0)$ for some $\delta>0$ since $V$ is strictly convex.
So let $\delta_1\in(-\delta_0,0)$.
With the right-hand side locked at $V'(0)-V'(\delta_1)>0$ we then adjust $\tau$ accordingly so that 
$H'(\delta_1)/\tau_1=V'(0)-V'(\delta_1)$.
We then set $u_{\tau_1}=\delta_1$.
It is clear that for $\tau>\tau_1$ there always exists a unique solution to this equation and furthermore $u^*_{\tau}$ is strictly increasing and bounded above by zero.
The limit has to be zero otherwise the continuity of $V'$ and $H'$ implies $\lim_{\tau\uparrow\infty}V'(u^*_{\tau})+H'(u^*_{\tau})/\tau=V'(\lim_{\tau\uparrow\infty}u^*_{\tau})<V'(0)$, a contradiction.
A similar analysis holds for $v>\theta$ and in this case $u^*_{\tau}$ converges to zero from above.
When $v=\theta$ then $u^*_{\tau}=0$ for all $\tau>0$ (i.e. it is a fixed point).
Analogous arguments hold for $k=V'(1)$: 
$H'(1)>0(<0)$ if and only if $v/\theta > \Upsilon(1)$ ($<\Upsilon(1)$) and $H'(1)=0$ if and only if $v/\theta=\Upsilon(1)$.
If $v/\theta>\Upsilon(1)$ ($<\Upsilon(1)$) then $u_\tau^*$ converges to $1$ from below (above) and 
when $v/\theta = \Upsilon(1)$, $u_{\tau}^*=1$ for all $\tau>0$.
\end{proof}
We now provide expansions for $u^*_{\tau}$ and the characteristic function $\Phi_{\tau,k,1/2}$.
Define the following quantities:
\begin{equation}\label{eq:v0v1apam}
\alpha_0:=
\frac{2\E^{-\kappa t}(v-\theta)\kappa}{\theta((2\kappa-\xi)^2+4\kappa\xi(1-\rho^2))}
,\quad
\alpha_1:=
\frac{2\E^{-\kappa t}(\kappa-\rho\xi)^2}{\kappa\theta ((2\kappa-\xi)^2+4\kappa\xi(1-\rho^2))}
(\theta\Upsilon(1) - v).
\end{equation}
The proofs are analogous to Lemma~\ref{lemma:a-dynamics} and~\ref{lemma:CharactNonSteepExp} 
and omitted.
Note that the asymptotics are in agreement with the properties of~$u^*_{\tau}(k)$ in Lemma~\ref{lem:v0v1ustarexist}.

\begin{lemma}\label{lem:v0v1}
Let $a\in\{0,1\}$ and assume that $v \ne \theta\Upsilon(a)$.
When $k=V'(a)$, the following expansions hold as~$\tau$ tends to infinity
(for some integer $s$):
\begin{align*}
u^*_\tau(k) & =  a + \alpha_a\tau^{-1} + \mathcal{O}\left(\tau^{-2}\right),\qquad
D(\tau,k) = \E^{\tau V'(a)(1-a)}\left(1+\mathcal{O}\left(\tau^{-1}\right)\right),\\
\Phi_{\tau,k,1/2}(u) & = \E^{-\frac{1}{2}u^2V''(a)}
\left(1+\left(\I \alpha_a u V''(a) - \frac{\I u^3 V'''(a)}{6}+\I u H'(a)\right)\tau^{-1/2}+\max(1,u^s)\mathcal{O}(\tau^{-1})\right).
\end{align*}
\end{lemma}


We now define the following functions from $\RR^*\times\{0,1\}$ to $\RR$ and then provide expansions for $F$ in ~\eqref{eq:DF}:
\begin{equation}\label{eq:varpiauxil}
\left\{
\begin{array}{ll}
\varpi_1(q,a) & :=
\E^{q^2 V''(a)/2}\pi \left[2 \Nn(q\sqrt{V''(a)})-1-\sgn(q)\right],\\
\varpi_2(q,a) & :=
-\sqrt{\frac{2\pi}{V''(a)}}+\E^{q^2 V''(a)/2}\pi q \left[1+\sgn(q)-2 \Nn\left(q\sqrt{V''(a)}\right)\right],\\
\varpi_3(q,a) & :=
\frac{\sqrt{2\pi}(q^2V''(a)-1)}{(V''(a))^{3/2}}
 - 2\pi q^2 |q| \exp\left(\frac{q^2V''(a)}{2}\right)\Nn\left(-|q|\sqrt{V''(a)}\right),\\
\varpi(q,a) & :=
\frac{\varpi_1(q)}{2\pi}+\frac{1}{2\pi \sqrt{\tau}}\left((a_1V''(a)+H'(a))\varpi_2(q,a)
+\frac{V'''(a) \varpi_3(q,a)}{6}\right).
\end{array}
\right.
\end{equation}
\begin{lemma}\label{lem:v0v1F}
Let $a\in\{0,1\}$ and assume that $v \ne \theta\Upsilon(a)$.
Then the following expansions hold as $\tau$ tends to infinity (with $a_0$ given in~\eqref{eq:v0v1apam}):
$$
F\left(\tau,V'(a),1/2\right)  = 
\frac{\ind_{\{a=1\}}\sgn(\alpha_1) - \ind_{\{a=0\}}\sgn(\alpha_0)}{2}
 - \frac{1}{\sqrt{2\pi \tau V''(a)}}
\left[1 + \sgn(a)\left(\frac{V'''(a)}{6V''(a)}-H'(a)\right)\right]\left[1+\mathcal{O}\left(\frac{1}{\tau}\right)\right].
$$
\end{lemma}
\begin{proof}
Consider the case $a=0$.
Set $P(u):=\I \alpha_0 u V''(0)-\I u^3 V'''(0)/6+\I u H'(0)$ and note that
$\overline{C(u,\tau,1/2)}:=\frac{1}{\left(-\I u-u^*_{\tau}\sqrt{\tau}\right)} - \frac{1}{\left(-\I u-u^*_{\tau}\sqrt{\tau}+\sqrt{\tau} \right)}$.
Using Lemma~\ref{lem:v0v1} and the definition of $F$ in~\eqref{eq:DF}:
\begin{equation}\label{eq:v0v1Fint}
F(\tau,V'(0),1/2) = 
\frac{1}{2\pi} \int_{\RR}\E^{-V''(0) u^2/2} \overline{C(u,\tau)} (1+P(u)\tau^{-1/2}+\mathcal{O}(\tau^{-1})) \D u.
\end{equation}
We cannot now simply Taylor expand $\overline{C(u,\tau,1/2)}$ for small $\tau$ and integrate term by term since in the limit $\overline{C(u,\tau,1/2)}$ is not $L^1$. 
This was the reason for introducing the time dependent term $u^*_{\tau}(V'(0))$ so that the Fourier transform exists for any $\tau>0$. 
Indeed, we easily see that $\overline{C(u,\tau,1/2)}=- \I/u + \mathcal{O}(\tau^{-1/2})$.
We therefore integrate these terms directly and then compute the asymptotics as $\tau$ tends to infinity.
Note first that since $|\overline{C(u,\tau,1/2)}|=\mathcal{O}(1)$, then
$\overline{C(u,\tau,1/2)} (1+P(u)\tau^{-1/2}+\mathcal{O}(\tau^{-1}))
=\overline{C(u,\tau,1/2)}(1+P(u)\tau^{-1/2})+\mathcal{O}(\tau^{-1})$.
Further for any $q\neq0$, $\int_{\RR}\E^{-V'' (0)u^2/2}\frac{1}{-\I u - q} \D u =\varpi_1(q,0)$,
$
\int_{\RR}\E^{-V''(0) u^2/2}\frac{\I u}{-\I u - q} \D u =\varpi_2(q,0)$
 and 
 $\int_{\RR}\E^{-V''(0) u^2/2}\frac{\I u^3}{-\I u - q} \D u =\varpi_3(q,0)$.
Now using the definition of $\varpi$ in~\eqref{eq:varpiauxil} and exchanging the integrals and the asymptotic (an analogous justification to the proof of Lemma~\ref{lemma:FourierAsymptotics}(i)) in~\eqref{eq:v0v1Fint} we obtain
$$
F(\tau,V'(0),1/2)
=\varpi\left(u^*_{\tau}\sqrt{\tau},0\right) - \varpi\left((u^*_{\tau}-1)\sqrt{\tau},0\right)
 + \mathcal{O}\left(\tau^{-1}\right).
$$
Using Lemma~\ref{lem:v0v1} and asymptotics of the cumulative normal distribution function we compute:
\begin{align*}
\varpi\left(u^*_{\tau}\sqrt{\tau},0\right) & = \varpi
\left(\alpha_0\tau^{-1/2}+\mathcal{O}\left(\tau^{-3/2}\right),0\right)
=-\frac{\sgn(\alpha_0)}{2}
-\frac{6H'(0)V''(0)-V'''(0)}{6\sqrt{2\pi}(V''(0))^{3/2}\sqrt{\tau}}
+\mathcal{O}\left(\tau^{-1}\right),\\
\varpi((u^*_{\tau}-1)\sqrt{\tau},0) & = 
\varpi\left(-\sqrt{\tau} + \alpha_0 \tau^{-1/2}+\mathcal{O}\left(\tau^{-3/2}\right),0\right)
=\frac{1}{\sqrt{2\pi V''(0) \tau}}
+\mathcal{O}\left(\tau^{-1}\right).
\end{align*}
The case $a=1$ is analogous using $\varpi(\cdot, 1)$ and the lemma follows.
\end{proof}

\begin{remark}
Consider $\Rf_{3b}$ and $\Rf_4$ with $k=V'(1)$ in Section~\ref{sec:asymmetricproofs}.
Here also $u^*_{\tau}(k)$ tends to $1$ and it is natural to wonder why we did not encounter the same issues with the limiting Fourier transform as we did in the present section.
The reason this was not a concern was that the speed of convergence ($\tau^{-1/2}$) of~$u^*_{\tau}$ to~$1$ 
was the same as that of the random variable $Z_{\tau,k,1/2}$ to its limiting value.
Intuitively the lack of steepness of the limiting lmgf was more important than any issues with the limiting Fourier transform.
In the present section steepness is not a concern, but again in the limit the Fourier transform is not defined. 
This becomes the dominant effect since $u^*_{\tau}(k)$ converges to~$1$
at a rate of~$\tau^{-1}$ while the re-scaled random variable~$Z_{\tau,k,1/2}$ 
converges to its limit at the rate~$\tau^{-1/2}$.
\end{remark}


\subsection{Forward smile asymptotics: Theorem~\ref{theorem:HestonLargeMatFwdSmileNonSteep}}\label{sec:fwdsmileproofnonsteeplargemat}

The general machinery to translate option price asymptotics into implied volatility asymptotics
has been fully developed by Gao and Lee~\cite{GL11}. 
We simply outline the main steps here. 
There are two main steps to determine forward smile asymptotics: 
(i) choose the correct root for the zeroth-order term in order to line up the domains (and hence functional forms) in Theorem~\ref{theorem:HestonLargeMatFwdSmileNonSteep} and Corollary~\ref{Cor:BSOptionLargeTime}; 
(ii) match the asymptotics. 

We illustrate this with a few cases from Theorem~\ref{theorem:HestonLargeMatFwdSmileNonSteep}.
Consider $\Rf_{3b}$ and $\Rf_4$ with $k>V'(1)$.
We have asymptotics for forward-start call option prices for $k>V'(1)$ in Theorem~\ref{theorem:HestonLargeMatFwdSmileNonSteep}.
The only BSM regime in Corollary~\ref{Cor:BSOptionLargeTime} where this holds is where 
$k\in(-\Sigma^2/2,\Sigma^2/2)$.
We now substitute our asymptotics for $\Sigma$ and at leading order we have the requirement: $k>V'(1)$ implies that 
$k\in(-v_0(k)/2,v_0(k)/2)$.
We then need to check that this holds only for the correct root $v_0$ used in the theorem.
Note that we only use the leading order condition here since if $k\in(-v_0(k)/2,v_0(k)/2)$ then there will always exist a $\tau_1>0$ such that $k\in(-v_0(k)/2+o(1),v_0(k)/2+o(1))$, for $\tau>\tau_1$.
Suppose now that we choose the root not as given in Theorem~\ref{theorem:HestonLargeMatFwdSmileNonSteep}.
Then for the upper bound we get the condition $kV(1)>0$. 
Since $V(1)<0$ we require $V'(1)<0$ and then this only holds for $V'(1)<k<0$.
This already contradicts $k>V'(1)$ but let's continue since it may be true for a more limited range of $k$.
The lower bound gives the condition
$(k-V(1))k>0$. 
But the upper bound implied that we needed $V'(1)<k<0$ and so further $k<V'(1)$.
Therefore $V'(1)<k<V(1)$ but this can never hold since simple computations show that $V'(1)>V(1)$.
Now let's choose the root according to the theorem.
For the upper bound we get the condition 
$
-\sqrt{(V^*-k)^2+k(V^*(k)-k)}< V^*(k)-k = -V(1)>0
$
and this is always true.
For the lower bound we get the condition $-\sqrt{(V^*-k)^2+k(V^*(k)-k)}< V^*(k)=k-V(1)$
and this is always true for $k>V'(1)$ since $V'(1)>V(1)$.
This shows that we have chosen the correct root for the zeroth-order term and we then simply match asymptotics for higher order terms.

As a second example consider $\Rf_2$ and $k>V'(u_{+}^*)$ in Theorem~\ref{theorem:HestonLargeMatFwdSmileNonSteep}.
Substituting the ansatz
$\sigma_{t,\tau}^2(k\tau)=v_0^{\infty}(k)+v_{1}^{\infty}(k,t)\tau^{-1/2}
 + v_{2}^{\infty}(k,t)\tau^{-1} + \mathcal{O}(\tau^{-3/2})$
into the BSM asymptotics for forward-start call options in Corollary~\ref{Cor:BSOptionLargeTime}, we find
$$
\EE\left(\E^{X_{\tau}^{(t)}}-\E^{k\tau}\right)^+
 = \exp\left(-\alpha_0^{\infty} \tau +\alpha_1^{\infty} \sqrt{\tau }+\alpha_2^{\infty}\right)
 \frac{4v_0^{3/2} }{\sqrt{2 \pi\tau} \left(4k^2-v_0^2\right)}\left(1+\mathcal{O}\left(\tau^{-1/2}\right)\right),
$$
where 
$\alpha_0^{\infty}:= \frac{k^2}{2v_0^{\infty}}-\frac{k}{2}+\frac{v_0^{\infty}}{8}$,
and
$\alpha_1^{\infty}:=  v_1^{\infty} \frac{4 k^2-v_0^2 }{8 v_0^2}$
and $\alpha_2^{\infty}$ is a constant, the exact value does not matter here.
We now equate orders with Theorem~\ref{theorem:largematasympcalls}.
At the zeroth order we get two solutions and since $V'(u^*_{+})>V(1)$, 
we choose the negative root such that matches the domains in Corollary~\ref{Cor:BSOptionLargeTime} and Theorem~\ref{theorem:largematasympcalls} for large~$\tau$ (using similar arguments as above).
At the first order we solve for $v_{1}^{\infty}$. 
But now at the second order, we can only solve for higher order terms if $\mu=1/2$ due to the term $\tau^{\mu/2-3/4} = \tau^{-1/2}$ in the forward-start option asymptotics in Theorem~\ref{theorem:largematasympcalls}. 
All other cases follow analogously.
\newpage


\appendix

\section{Proof of Lemma~\ref{lem:nonsteeplem}}\label{sec:prooflargematfourtransf}

Define the function $C_{\tau,k,\alpha}:\RR\to\mathbb{C}$ by
\begin{equation}\label{eq:Cdeflargetime}
C_{\tau,k,\alpha}(u):=\frac{\tau^{\alpha}}{(u+\I \tau^{\alpha}(u^*_{\tau}-1)(u+\I \tau^{\alpha}u^*_{\tau})},
\end{equation}
with its conjugate given in~\eqref{eq:CConj}.

\begin{lemma}\label{lem:L1lem}
There exists $\tau^*_0>0$ such that
$\int_{\RR}|\Phi_{\tau,k,\alpha}(u) \overline{C_{\tau,k,\alpha}(u)}|  \D u 
<\infty$
for all $\tau>\tau^*_0$, $k\in\mathcal{A}$, $u^*_{\tau}(k)\not\in\{0,1\}$.
\end{lemma}
\begin{proof}
We compute:
\begin{align}
\nonumber
\int_{\RR}\left| \Phi_{\tau,k}(u) \overline{C_{\tau,k,\alpha}(u)}  \right|  \D u 
&=
\int_{|u|\leq \tau^{\alpha}}\left| \Phi_{\tau,k,\alpha}(u) \overline{C_{\tau,k,\alpha}(u)}  \right|  \D u 
+\int_{|u|>\tau^{\alpha}}\left| \Phi_{\tau,k,\alpha}(u) \overline{C_{\tau,k,\alpha}(u)}  \right|  \D u  \\ \label{eq:finiteineq}
&\leq \frac{2\tau^{-\alpha}}{|u^*_{\tau}(k)(u_{\tau}^*(k)-1)|} \int_{|u|\leq \tau^{\alpha}}\left| \Phi_{\tau,k,\alpha}(u) \right| \D u 
+\int_{|u|>1}\frac{\D u}{u^2},
\end{align}
where the inequality follows from the simple bounds
$$
\left|\overline{C_{\tau,k,\alpha}(u)}  \right|\leq \frac{\tau^{-2\alpha}}{|u^*_{\tau}(k)(u_{\tau}^*(k)-1)|},\text{ for all }|u| \leq \tau^{\alpha},
\qquad
\left|\overline{C_{\tau,k,\alpha}(u)}  \right|\leq \frac{\tau^{\alpha}}{u^2}
\qquad
\text{and}\qquad
|\Phi_{\tau,k,\alpha}|\leq1.
$$
Finally~\eqref{eq:finiteineq} is finite since $u^*_{\tau}(k)\neq 1$, $u^*_{\tau}(k)\neq 0$.
\end{proof}

We denote the convolution of two functions $f,h\in L^1(\RR)$ by $(f\ast g)(x):=\int_{\RR}f(x-y)g(y) \D y$,
and recall that $(f\ast g)\in L^1(\RR)$.
For $f\in L^1(\RR)$, we denote its Fourier transform by
$(\mathcal{F}f)(u):=\int_{\RR}\E^{\I u x}f(x) \D x$ 
and the inverse Fourier transform by $
(\mathcal{F}^{-1}h)(x):=\frac{1}{2\pi}\int_{\RR}\E^{-\I u x}h(u) \D u.
$ 
For $j=1,2,3$, define the functions~$g_j:\RR_+^2\to\RR_+$ by
$$
g_j(x,y):=\left\{ 
  \begin{array}{ll}
(x-y)^+, \quad & \text{if } j=1, \\
(y-x)^+,\quad & \text{if } j=2,\\
 \min(x,y),\quad & \text{if } j=3.
 \end{array} \right.
$$
and define~$\widetilde{g}_j:\RR\to\RR_+$ by
$\label{eq:gjtilde}
\widetilde{g}_j(z) := \exp\left(-u^*_{\tau}(k) z \tau^{\alpha}\right)g_j(\E^{z \tau^{\alpha}},1)$.
Recall the $\QQ_{k,\tau}$-measure defined in~\eqref{eq:MeasureChange} 
and the random variable~$Z_{k,\tau,\alpha}$ defined on page~\pageref{eq:ztaukalpha}. 
We now have the following result:
\begin{lemma}~\label{lem:optpricerepnonsteep}
There exists $\tau^*_1>0$ such that for all $k\in\mathcal{A}$ and $\tau>\tau^*_1$:
\begin{align}~\label{eq:Parslargetime}
\EE^{\QQ_{k,\tau}}\left[\widetilde{g}_j(Z_{k,\tau,\alpha})\right]
= \begin{dcases*}
       \frac{1}{2\pi}\int_{\RR} \Phi_{\tau,k,\alpha}(u)\overline{C_{\tau,k,\alpha}(u)} \D u,  & if $j=1,u_{\tau}^*(k)>1$,\\
       \frac{1}{2\pi}\int_{\RR}\Phi_{\tau,k,\alpha}(u)\overline{C_{\tau,k,\alpha}(u)} \D u,  & if $j=2, u_{\tau}^*(k)<0$, \\
  -\frac{1}{2\pi}\int_{\RR}\Phi_{\tau,k,\alpha}(u)\overline{C_{\tau,k,\alpha}(u)} \D u,  & if $j=3, 0<u_{\tau}^*(k)<1. $
        \end{dcases*}
\end{align}
\end{lemma}
\begin{proof}
Assuming (for now) that $\widetilde{g}_j\in L^1(\RR)$, we have for any $u\in\RR$,
$
\left(\mathcal{F}\widetilde{g}_{j}\right)(u)
:=\int_{\RR}\widetilde{g}_j(z)\E^{\I u z}\D z,
$
for $j=1,2,3$.
For $j=1$ we can write
$$
\int_0^{\infty}\E^{-u^*_{\tau}z\tau^{\alpha}}\left(\E^{z\tau^{\alpha}}-1\right)\E^{\I u z}\D z
=\left[\frac{\E^{z\left(\I u-u^*_{\tau}\tau^{\alpha}+\tau^{\alpha}\right)}}
{\left(\I u-u^*_{\tau}\tau^{\alpha}+\tau^{\alpha}\right)}\right]_{0}^{\infty}
-\left[\frac{\E^{z\left(\I u-u^*_{\tau}\tau^{\alpha}\right)}}
{\left(\I u-u^*_{\tau}\tau^{\alpha}\right)}\right]_{0}^{\infty} 
=C_{\tau,k,\alpha}(u),
$$
which is valid for $u^*_{\tau}(k)>1$ with $C_{\tau,k,\alpha}$ in~\eqref{eq:Cdeflargetime}. 
For $j=2$ we can write
$$
\int_{-\infty}^0\E^{-u^*_{\tau}z\tau^{\alpha}}\left(1-\E^{z\tau^{\alpha}}\right)\E^{\I u z}\D z
=\left[\frac{\E^{z\left(\I u-u^*_{\tau}\tau^{\alpha}\right)}}
{\left(\I u-u^*_{\tau}\tau^{\alpha}\right)}\right]_{-\infty}^{0}-\left[\frac{\E^{z\left(\I u-u^*_{\tau}\tau^{\alpha}+\tau^{\alpha}\right)}}
{\left(\I u-u^*_{\tau}\tau^{\alpha}+\tau^{\alpha}\right)}\right]_{-\infty}^{0} 
=C_{\tau,k,\alpha}(u),
$$
which is valid for $u_{\tau}^*(k)<0$. Finally, for $j=3$ we have
\begin{align*}
\int_{\RR}\E^{-u^*_{\tau}z\tau^{\alpha}}\left(\E^{z\tau^{\alpha}}\wedge1\right)\E^{\I u z}\D z
&=\int_{-\infty}^{0}\E^{-u^*_{\tau}z\tau^{\alpha}}\E^{z\tau^{\alpha}}\E^{\I u z}\D z
+\int_{0}^{\infty}\E^{-u^*_{\tau}z\tau^{\alpha}}\E^{\I u z}\D z \\
&=\left[\frac{\E^{z\left(\I u-u^*_{\tau}\tau^{\alpha}+\tau^{\alpha}\right)}}
{\left(\I u-u^*_{\tau}\tau^{\alpha}+\tau^{\alpha}\right)}\right]_{-\infty}^{0}
+\left[\frac{\E^{z\left(\I u-u^*_{\tau}\tau^{\alpha}\right)}}
{\left(\I u-u^*_{\tau}\tau^{\alpha}\right)}\right]_{0}^{\infty} 
=-C_{\tau,k,\alpha}(u),
\end{align*}
which is valid for $0<u^*_{\tau}(k)<1$. 
From the definition of 
the $\QQ_{k,\tau}$-measure in~\eqref{eq:MeasureChange} and the random variable $Z_{k,\tau,\alpha}$ on page~\pageref{eq:ztaukalpha} we have
$$
\EE^{\QQ_{k,\tau}}\left[\widetilde{g}_j(Z_{\tau,k,\alpha})\right]
=\int_{\RR}q_j(k\tau^{1-\alpha}-y)p(y) \D y = (q_j\ast p)(k\tau^{1-\alpha}),
$$
with $q_j(z)\equiv\widetilde{g}_j(-z)$ and $p$ denoting the density of $X_{\tau}^{(t)} \tau^{-\alpha}$.
On the strips of regularity derived above we know there exists $\tau_0>0$ such that
$q_j\in L^1(\RR)$ for $\tau>\tau_0$.
Since $p$ is a density, $p\in L^1(\RR)$, and therefore 
\begin{equation}\label{eq:conv}
\mathcal{F}(q_j\ast p)(u)=\mathcal{F}q_j(u) \mathcal{F}p(u).
\end{equation}
We note that $\mathcal{F}q_j(u)\equiv\mathcal{F}\widetilde{g}_j(-u)\equiv\overline{\mathcal{F}\widetilde{g}_j(u)}$ and hence
\begin{equation}\label{eq:simpconv}
\mathcal{F}q_j(u) \mathcal{F}p(u)
\equiv\E^{\I u k\tau^{1-\alpha}}\Phi_{\tau,k,\alpha}(u) \overline{C_{\tau,k,\alpha}(u)}.
\end{equation}
Thus by Lemma~\ref{lem:L1lem} there exists $\tau_1>0$ 
such that $\mathcal{F}q_j \mathcal{F}p\in L^1(\RR)$ for $\tau>\tau_1$. 
By the inversion theorem~\cite[Theorem 9.11]{R87} this then implies from~\eqref{eq:conv} and~\eqref{eq:simpconv} that for $\tau>\max(\tau_0,\tau_1)$:
\begin{align*}
\EE^{\QQ_{k,\varepsilon}}\left[\widetilde{g}_j(Z_{\tau,k,\alpha})\right]
&= (q_j\ast p)(k\tau^{1-\alpha}) 
=\mathcal{F}^{-1}\left(\mathcal{F}q_j(u) \mathcal{F}p(u)\right)(k\tau^{1-\alpha}) \\ 
&=\frac{1}{2\pi}\int_{\RR} \E^{-\I u k\tau^{1-\alpha}}\mathcal{F}q_j(u) \mathcal{F}p(u) \D u
= \frac{1}{2\pi}\int_{\RR} \Phi_{\tau,k,\alpha}(u) \overline{C_{\tau,k,\alpha}(u)}  \D u .
\end{align*}
\end{proof}

We now move onto the proof of Lemma~\ref{lem:nonsteeplem}.
We use our time-dependent change of measure defined in~\eqref{eq:MeasureChange} to write our forward-start option price for $j=1,2,3$ as
$$
\EE\left(g_j(\E^{X_{\tau}^{(t)}},\E^{k\tau})\right)
= \E^{-\tau\left[ku^*_{\tau}(k)-\Lambda^{(t)}_{\tau}\left(u^*_{\tau}(k)\right)\right]}\E^{k\tau}\EE^{\QQ_{k,\tau}}\left[\widetilde{g}_j(Z_{\tau,k,\alpha})\right],
$$
with $Z_{\tau,k,\alpha}$ defined on page~\pageref{eq:ztaukalpha}.
We now apply Lemma~\ref{lem:optpricerepnonsteep} and then convert to forward-start call option prices
using Put-Call parity and that in the Heston model $(\E^{X_t})_{t\geq0}$ is a true martingale~\cite[Proposition 2.5]{AP07}.
Finally the expansion for $\exp\left({-\tau\left(k (u^*_{\tau}(k)-1)-\Lambda^{(t)}_{\tau}(u^*_{\tau}(k))\right)}\right)$
follows from Lemma~\ref{lem:fwdmgflargetauexpansion}.


\section{Tail Estimates}

\begin{lemma}\label{lem:expsmalllargetime}
There exists $\beta>0$ such that the following tail estimate holds for all $k \in \mathcal{A}$ and $u^*_{\tau}(k)\not\in\{0,1\}$
as $\tau$ tends to infinity:
$
\left| \int_{|u|>\tau^{\alpha}} \Phi_{\tau,k,\alpha}(u) \overline{C_{\tau,k,\alpha}(u)} \D u \right| = \mathcal{O}(\E^{-\beta \tau}).
$
\end{lemma}
\begin{proof}
By the definition of $\Phi_{\tau,k,\alpha}$ in~\eqref{eq:CharacNonSteep} we have
$
|\Phi_{\tau,k,\alpha}(z \tau^{\alpha})|=\exp\left(\tau(\Re [\Lambda_{\tau}^{(t)}(\I z+u^*_\tau)]-\Lambda_{\tau}^{(t)}(u^*_\tau))\right).
$
For  $|z|>1$ we have the simple estimate 
$
\left| \overline{C_{\tau,k,\alpha}(z\tau^{\alpha})} \right|
\leq \tau^{-\alpha}/z^2,
$
and therefore
$$
\left| \int_{|u|>\tau^{\alpha}} \Phi_{\tau,k,\alpha}(u) \overline{C_{\tau,k,\alpha}(u)} \D u \right|
\leq
\tau^{\alpha}\int_{|z|>1}\left| \Phi_{\tau}(z \tau^{\alpha}) \right| \left|  \overline{C_{\tau,k,\alpha}(z \tau^{\alpha})}  \right| \D z
\leq
  \int_{|z|>1}\E^{\tau(\Re [\Lambda_{\tau}^{(t)}(\I z+u^*_\tau)]-\Lambda_{\tau}^{(t)}(u^*_\tau))} \frac{\D z}{z^2},
$$
for all $\tau>0$.
We deal with the case $z>1$. Analogous arguments hold for the case $z<-1$. 
Lemma~\ref{lem:largetimesaddlefwdmgf}(i) implies that there exists $\tau_1$ such that for $\tau>\tau_1$: 
$$
\int_{z>1}\E^{\tau(\Re [\Lambda_{\tau}^{(t)}(\I z+u^*_\tau)]-\Lambda_{\tau}^{(t)}(u^*_\tau))} \frac{\D z}{z^2}
\leq
\E^{\tau(\Re [V(\I+u^*_\tau)]-V(u^*_\tau))+\mathcal{O}(1)} \int_{z>1} \frac{\D z}{z^2}.
$$
Using Lemma~\ref{lem:largetimesaddlefwdmgf}(ii) we compute
$$
\Re \Lambda_{\tau}^{(t)}(\I+u^*_\tau)-\Lambda_{\tau}^{(t)}(u^*_\tau)
=\Re V(\I+u^*_\tau)-V(u^*_\tau)
+(\Re H(\I+u^*_\tau)-H(u^*_\tau))/\tau
+\mathcal{O}(\tau^{-n}),
$$
for any $n>0$.
Now using that $V$ and $H$ are continuous and Assumption~\ref{assump:ustartaueq} we have that
$
\Re V(\I+u^*_\tau)-V(u^*_\tau)=\Re V(\I+u_{\infty})-V(u_{\infty})+o(1)
$
and
$
\Re H(\I+u^*_\tau)-H(u^*_\tau)=\Re H(\I+u_{\infty})-H(u_{\infty})+o(1),
$
as $\tau$ tends to infinity.
Lemma~\ref{lem:largetimesaddlefwdmgf}(iii) implies that $\Re V(\I+u_{\infty})-V(u_{\infty})<0$ and the lemma follows.
\end{proof}

\begin{lemma}\label{lem:largetimesaddlefwdmgf} \
\begin{enumerate}[(i)]
\item
The expansion $ \exp(\Lambda_{\tau}^{(t)}(\I z+u^*_{\tau}))= \exp(V(\I z+u^*_{\tau})+ H(\I z+u^*_{\tau})\tau^{-1})\mathcal{R}(\tau)$ holds as $\tau$ tends to infinity where $\mathcal{R}(\tau)=\E^{\mathcal{O}(\E^{-\beta\tau})}$ for some $\beta>0$ and $\mathcal{R}$ is uniform in $z$.
\item
There exists $\tau_1^*$ such that $\Re \Lambda_{\tau}^{(t)}(\I z+u^*_{\tau})\leq \Re \Lambda_{\tau}^{(t)}(\I \sgn(z)+u^*_{\tau})$ for all $z>|1|$ and $\tau>\tau_1^*$.
\item
For all $a\in\mathcal{D}_{\infty}^{o}$ the function $\RR\ni z \mapsto \Re V(\I z+a)$ 
has a unique maximum at zero.
\end{enumerate}
\end{lemma}
\begin{proof}\
\begin{enumerate}[(i)]
\item
The proof of the expansion follows from Assumption~\ref{assump:ustartaueq}
and analogous steps to the proofs of Lemma~\ref{lem:fwdmgflargetauexpansion} and Lemma~\ref{lem:charactsimp}.
The proof of uniformity of the remainder $\exp\left( \Lambda_{\tau}^{(t)}(\I z + a) -V(\I z + a)-H(\I z + a)\tau^{-1}\right)$ in $z$ involves tedious but straightforward computations and is omitted for brevity.
See Figure~\ref{fig:hestlargetimetail}(a) for a visual illustration.
\item
Assumption~\ref{assump:ustartaueq} implies that there exists $\tau_1^*$ such that $u_{\tau}^*\in\mathcal{D}_{\infty}^{o}$ for all $\tau>\tau_1^*$.
So we need only show that for all $\tau>0$ and $a\in\mathcal{D}_{t,\tau}^{o}$:
$\Re \Lambda_{\tau}^{(t)}(\I z+a)\leq \Re \Lambda_{\tau}^{(t)}(\I \sgn(z)+a)$ for all $z>|1|$.
The proof of this result involves tedious but straightforward computations and is omitted for brevity.
See Figure~\ref{fig:hestlargetimetail}(b) for a visual illustration.
\item
The proof of (iii) is straightforward and follows the same steps as~\cite[Appendix C]{JR12}.
We omit it for brevity.
\end{enumerate}
\end{proof}
\begin{figure}[h!tb] 
\centering
\mbox{\subfigure[]{\includegraphics[scale=0.6]{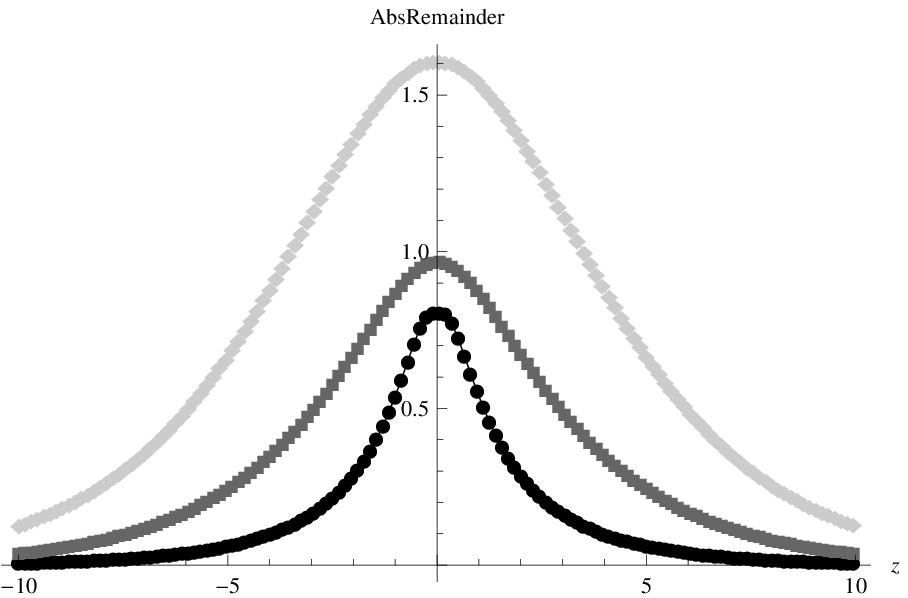}}\quad
\subfigure[]{\includegraphics[scale=0.6]{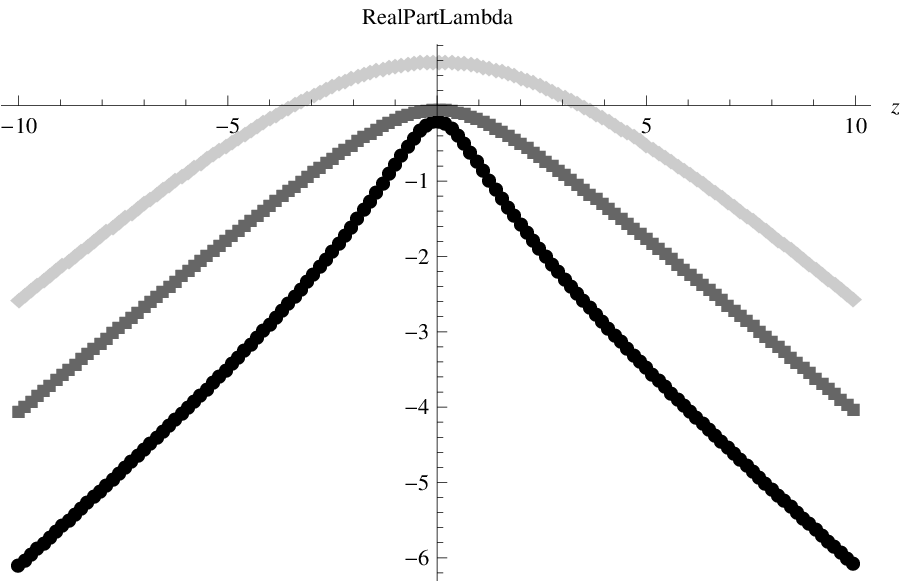}}}
\caption{On the left we plot the map $z\mapsto \left|\exp\left( \Lambda_{\tau}^{(t)}(\I z + a) -V(\I z + a)-H(\I z + a)\tau^{-1}\right)\right|$ and on the right we plot the map $z\mapsto \Re \Lambda_{\tau}^{(t)}(\I z + a)$. 
Here $a=-3$ (circles), $a=0.5$ (squares) and $a=3$ (diamonds) and the 
parameters are the same as Figure 2 with $t=1$ and $\tau=5$.}
\label{fig:hestlargetimetail}
\end{figure}

\end{document}